\documentclass[prd,nofootinbib,floats,aps,twocolumn,tightenlines,superscriptaddress]{revtex4-2}
\usepackage{natbib}
\usepackage{graphicx}
\usepackage{mathtools}
\usepackage[normalem]{ulem}
\usepackage{hyperref}
\hypersetup{
	colorlinks=true,
	linkcolor=blue,
	filecolor=magenta,      
	urlcolor=cyan,
}
\usepackage{commath}
\usepackage{bm}
\usepackage{amsfonts,amsmath,amssymb,amsthm}
\usepackage[usenames,dvipsnames]{xcolor}
\usepackage{dsfont}
\usepackage{enumitem}
\usepackage{color}
\usepackage{tikz}
\usepackage{physics}
\usepackage{verbatim}
\newcommand{\ii}{\mathrm{i}}
\newcommand{\proj}[2]{| {#1} \rangle\!\langle {#2} |}

\renewcommand*\d[2][]{%
	\mathrm{d}%
	\ifx\relax#1\relax\else
	\rule{-0.02em}{1.5ex}^{#1}\rule{0.08em}{0ex}\!
	\fi
	#2\,
}

\newcommand{\phih}{\hat \phi}

\newcommand{\muh}{\hat{\mu}}

\newcommand{\rhoh}{\hat{\rho}}
\newcommand{\maf}{\mathsf}
\newcommand{\prt}{\mathcal{P}}
\newcommand{\geno}{\hat{\mathcal{O}}}
\newcommand{\hilb}{\mathcal{H}}

\makeatletter 

\renewcommand\onecolumngrid{
\do@columngrid{one}{\@ne}%
\def\set@footnotewidth{\onecolumngrid}
\def\footnoterule{\kern-6pt\hrule width 1.5in\kern6pt}%
}

\renewcommand\twocolumngrid{
        \def\footnoterule{
        \dimen@\skip\footins\divide\dimen@\thr@@
        \kern-\dimen@\hrule width.5in\kern\dimen@}
        \do@columngrid{mlt}{\tw@}
}%
\makeatother

\newcommand\restr[2]{{
		\left.\kern-\nulldelimiterspace 
		#1 
		\vphantom{\normal|} 
		\right|_{#2} 
}}

\newcommand{\diff}{\mathrm{d}}

\newcommand{\R}[1]{\mathbb{R}^{#1}}
\newcommand{\C}[1]{\mathbb{C}^{#1}}
\newcommand{\beq}{\begin{equation}}
	\newcommand{\eeq}{\end{equation}}

\newtheorem{lemma}{Lemma}
\newtheorem*{claim*}{Claim}

\begin{document}
	
\title{Non-perturbative method for particle detectors with continuous interactions}
	
\author{Jos\'{e} Polo-G\'{o}mez}
\email{jpologomez@uwaterloo.ca}
\affiliation{Department of Applied Mathematics, University of Waterloo, Waterloo, Ontario, N2L 3G1, Canada}
\affiliation{Institute for Quantum Computing, University of Waterloo, Waterloo, Ontario, N2L 3G1, Canada}
\affiliation{Perimeter Institute for Theoretical Physics, Waterloo, Ontario, N2L 2Y5, Canada}
	
\author{Eduardo Mart\'{i}n-Mart\'{i}nez}
\email{emartinmartinez@uwaterloo.ca}
\affiliation{Department of Applied Mathematics, University of Waterloo, Waterloo, Ontario, N2L 3G1, Canada}
\affiliation{Institute for Quantum Computing, University of Waterloo, Waterloo, Ontario, N2L 3G1, Canada}
\affiliation{Perimeter Institute for Theoretical Physics, Waterloo, Ontario, N2L 2Y5, Canada}

\begin{abstract}
		
We show that detector switching profiles consisting of trains of delta couplings are a useful computational tool to efficiently approximate results involving continuous switching functions, both in setups involving a single detector and multiple ones. The rapid convergence to the continuous results at all orders in perturbation theory for sufficiently regular switchings means that this tool can be used to obtain non-perturbative results for general particle detector phenomena with continuous switching functions.  

		
\end{abstract}
	
\maketitle
	
\section{Introduction}\label{Section: introduction}
	
Performing local operations on quantum fields is a challenging endeavour, both from the technical and the foundational point of view~\cite{Hellwig1969,Hellwig1970formal,Hellwig1970operations,Sorkin1993,FewsterVerch,Borsten2021,Jubb2021,PDfromLocalizedQFTs}. Among the several approaches used to implement operations and local measurements of a quantum field, particle detectors---non-relativistic quantum systems locally coupled to quantum fields---have succeeded as models of local probes in quantum field theory (QFT), and more specifically in relativistic quantum information (RQI). Particle detector models have allowed us to make physical sense and build intuition for phenomena ranging from the concept of particle~\cite{Unruh,UnruhWald,Redhead1995,Malament1996,Colosi2008,Papageorgiou2019} to the entanglement structure of QFTs~\cite{Valentini1991,Reznik2003,Reznik2005,Silman2007,Pozas2015,Salton2015,Pozas2016,Henderson2019,Henderson2020,Erickson2021notreallyharvesting,Erickson2021fallingintoBH,Stritzelberger2021,Robbins2021,Foo2021,MendezAvalos2022,Henderson2022,Bueley2022,Peropadre2010,Janzen2023,Gooding2023,Hotta2015,Trevison2018,Trevison2019,EntanglementStructureProbes,FullyRelativisticEH}. They also have provided a basis for modelling quantum information protocols in relativistic setups (e.g.,~\cite{Cliche2009,Landulfo2016,Jonsson2017,Jonsson2018,Simidzija2020,Barcellos2021,Erickson2022rapidchannel,Erickson2022teleportation,Koji2023,Jonsson2015,Blasco2016,Blasco2015,Hotta2008,Hotta2011,Hotta2014,Funai2017}), as well as for the formulation of a measurement theory consistent with relativity~\cite{MeasurementTheory,MeasurementsEH,DanIreneML}.

Typically, particle detectors are coupled to a quantum field with some strength (determined by a parameter $\lambda$), and for a specific period of time (determined by some switching function $\chi$). Depending on the specific configuration, particle detectors can be used either as emitters and receptors of information, or as measurement devices, i.e., local probes that gather information about the field through their interaction. In both cases, in order to make predictions we need to be able to compute the final state of the detectors after their interaction with the field. However, the full theory including the interaction between the detector and the field cannot be exactly solved in general, as is the case for most interacting field theories. 

A frequent avenue to circumvent the complications mentioned above is the use of perturbation theory. In this approach, the final state of the detectors is expanded in a series of terms proportional to increasingly higher powers of the coupling parameter $\lambda$. Under the assumption that the detectors are weakly coupled to the field, the higher order contributions can be neglected, leaving a truncated series as the final result. Despite its popularity, the perturbative approach has its limitations. To start with, a proof that the perturbative series is convergent in all regimes is still missing for particle detector models. Even when we are in a regime of coupling strengths for which the perturbative series converges, the truncated result is an approximation whose error is not known a priori, and can only be trusted to be accurate for ``small enough'' coupling parameters. Furthermore, there are many physical situations where one needs to go beyond leading order in perturbation theory, which can be technically challenging. 
Finally, non-perturbative effects in the model cannot in general be understood using perturbative methods.


These limitations call for the use of non-perturbative techniques. One way to proceed is to work in specific scenarios where the evolution of the detectors and the field can be calculated exactly. One of those scenarios is found in the context of continuous-variable quantum mechanics (see, e.g.,~\cite{GQMRev,Serafini2017} for reviews). Specifically, when the quantum field has both infrared and ultraviolet cutoffs, it can be reduced to (or approximated by) a lattice, so that if the detectors and the field are initially in a so-called \textit{Gaussian state} (i.e., the Wigner function describing their joint state is Gaussian), and if their coupling is linear, then the evolution can be solved exactly using Gaussian methods~\cite{Brown2013,Bruschi2013}. Another scenario where calculations can be carried out non-perturbatively is when the coupling between the detectors and the field consists of sudden interactions, i.e., one or several delta couplings~\cite{PozasKerstjens2017,Jonsson2018,TalesAhmed2022,Erickson2022rapidchannel,Erickson2022teleportation,Erickson2022Fermi,Kollas2023}. In this case, as we will review in Secs.~\ref{Subsection: Setup single detector} and~\ref{Subsection: Setup multiple detectors}, the evolution can always be written as a sequence of known unitaries, and, moreover, these unitaries can often be exactly evaluated (see, e.g.,~\cite{Hotta2008,Simidzija2017nonperturbative,Sahu2022,Erickson2023Gaussian}). 


In this paper we devise an efficient non-perturbative method that applies to particle detectors that couple to the quantum field continuously in time. Given the physical relevance of these scenarios, and the fact that so far they have been mostly analyzed perturbatively, these techniques may allow us to explore possible new non-perturbative phenomena in the study of the dynamics induced by time-extended couplings, such as, e.g., in the Unruh effect or entanglement harvesting, among many others.

The paper is organized as follows: in Secs.~\ref{Section: Single detector} and~\ref{Section: Multiple detectors} we develop, respectively, the non-perturbative method for setups involving a single and multiple particle detectors. For each case, we first review the formalism of delta coupling interactions, we then proceed to show how they can be used to approximate the dynamics of continuous couplings, and conclude stating the approximation result that guarantees the efficiency of the method, whose proof can be found in Appendices~\ref{Appendix: Proof of the approximation result} and~\ref{Appendix: The n-dimensional Riemann-Stieltjes integral}. In Sec.~\ref{Section: Examples}, we demonstrate the method by applying it to specific scenarios involving one and two detectors. Finally, our conclusions are presented in Sec.~\ref{Section: Conclusion}.

\section{Single detector case}\label{Section: Single detector}

In this section we describe the setup of a single particle detector interacting with a quantum field through a sequence of delta couplings, and we explain how it can be used to approximate the phenomenology of setups involving continuous detector-field couplings.  

\subsection{Setup}\label{Subsection: Setup single detector}

We consider a localized non-relativistic quantum system (i.e., a particle detector~\cite{Unruh,DeWitt,UnruhWald,Tales2022PD}), moving in a general (1+$d$)-dimensional globally hyperbolic spacetime $(\mathcal{M},\maf g)$, whose centre of mass follows a trajectory $\maf x (\tau)$ parametrized by its proper time $\tau$. The detector is modelled as a quantum system with free Hamiltonian $\hat H_{\text{free}}$. We also consider that we are within the regime in which the detector can be assumed to be Fermi-Walker rigid to a good degree of approximation~\cite{Poisson2004}, and that, therefore, in the reference frame defined by a set of Fermi normal coordinates for the detector's centre of mass $\maf z = (\tau,\bm z)$, its shape is kept constant. This assumption is commonly made in particle detector models~\cite{Schlicht2004,Louko2006,Langlois,TalesBruno2020,Tales2022PD}, and it is regarded as a physically realistic approximation in many experimental setups~\cite{Kazantsev1974,Edu2018}.

The detector is weakly coupled to a quantum field $\phih$, and its coupling can be described in the interaction picture by the Hamiltonian weight\footnote{Notice that the Hamiltonian weight is a scalar, which upon multiplication by the geometric factor $\sqrt{-g(\maf z)}$ yields the Hamiltonian \textit{density}, where $g(\maf z)$ is the determinant of the metric at $\maf z$.}~\cite{TalesBruno2020}
\begin{equation}\label{Eq: interaction Hamiltonian}
\hat h_\text{int}(\tau, \bm z) = \lambda \chi(\tau) \big[ F(\bm z) \muh^\dagger_{\alpha}(\tau) \,\geno^{\alpha}(\tau,\bm z) + \text{H.c.}\big].
\end{equation}
Here, $\lambda$ is the coupling strength of the interaction, and $\chi$ and $F$ are the switching and smearing functions that modulate the coupling in time and space, where in particular $F$ depends on the shape of the detector (and, in general, can be complex~\cite{QOptics,Edu2013,Pozas2016,Lopp2020,Tales2022PD}). $\muh^{\alpha}$ and $\geno^{\alpha}$ are arbitrary tensor operators of the detector and the field, respectively, and $\alpha$ is a general multi-index (for example made of several spacetime indices). This is the most general interaction Hamiltonian weight that we can write for a Fermi-Walker rigid particle detector coupled to a field. In particular, for a monopolar coupling with a real scalar field one recovers the Unruh-DeWitt model, but this interaction also covers multipolar couplings with real/complex, scalar/vector/tensor, bosonic/fermionic fields. 

In the case of multiple sudden interactions, the detector couples to the field via a train of delta couplings at a sequence of times \mbox{$\{\tau_j, \, j=1,\hdots,N\}$, $N\geq1$}, and the switching function is thus given by
\begin{equation}\label{Eq: train of deltas}
\chi(\tau) = \sum_{j=1}^N \eta_j \,\delta(\tau-\tau_j),
\end{equation}
where $\eta_j$ is the strength with which the detector couples to the field at time $\tau_j$. With the interaction Hamiltonian weight given by Eq.~\eqref{Eq: interaction Hamiltonian}, the evolution operator in the interaction picture is given by
\begin{align}\label{Eq: evolution operator}
\hat U \!& =\! \mathcal{T}_\tau\!\exp\!\bigg( \!\!-\!\ii\!\int\!\diff\maf z\, \hat h_\text{int}(\maf z)   \bigg) \\
& = \!\mathcal{T}_\tau\! \exp\!\bigg\{\!\! -\!\ii\lambda\!\sum_{j=1}^{N}\! \eta_j\! \bigg[ \muh_{\alpha,j}^\dagger\!\! \int\!\!\diff\bm z \sqrt{-g_j} F(\bm z)\geno_j^\alpha(\bm z)  
 \!+\! \text{H.c.}\bigg]\!\bigg\} \nonumber\\
& = \!\mathcal{T}_\tau\!\! \prod_{j=1}^{N}\! \exp\!\bigg\{\!\! -\!\ii\lambda\eta_j\!\bigg[\muh_{\alpha,j}^\dagger \!\!\int\!\!\diff \bm z \sqrt{-g_j} F(\bm z)\geno^\alpha_j(\bm z)\!+\!\text{H.c.}\bigg]  \!\bigg\}, \nonumber
\end{align}
where \mbox{$g_{j} \equiv \det g_{\mu\nu}(\tau_j,\bm z) $} is the determinant of the spacetime metric at $(\tau_j,\bm z)$, and similarly $\muh_{\alpha,j}^\dagger \equiv \muh_\alpha^\dagger(\tau_j)$, and \mbox{$\geno^\alpha_j(\bm z)\equiv\geno^\alpha(\tau_j,\bm z)$}. $\mathcal{T}_\tau\exp$ and $\mathcal{T}_\tau$ respectively denote the time-ordered exponential and the time ordering operation with respect to the detector's proper time $\tau$, which is an acceptable choice of time-ordering parameter as long as we are in the regime of validity of the detector approximation, where the detector model is effectively covariant~\cite{TalesBruno2021}. In particular, writing 
\begin{equation}\label{Eq: each unitary single detector}
\!\!\hat U_j\! \coloneqq  \exp\!\bigg\{\!\! -\!\nobreak\hspace{0.09em}\ii\lambda\eta_j \bigg[ \muh_{\alpha,j}^\dagger\!\! \int\!\!\diff \bm z\sqrt{-g_j} F(\bm z)\geno^{\alpha}_j(\bm z)\!\nobreak\hspace{0.12em} +\!\nobreak\hspace{0.15em} \text{H.c.}\bigg] \! \bigg\},
\end{equation}
we can write
\begin{equation}\label{Eq: sequence of unitaries}
\hat U = \hat U_{N} \cdots \hat U_1.
\end{equation}
Notice that the $\hat U_{j}$ operators do not commute with each other, but the action of the time ordering operation in Eq.~\eqref{Eq: evolution operator} yields Eq.~\eqref{Eq: sequence of unitaries}. After the interaction, the joint state of the detector and the field is
\begin{equation}
\rhoh = \hat U \rhoh_0 \hat U^\dagger,
\end{equation}
where $\rhoh_0$ is the initial detector-field state. In particular, the final state of the detector after the interaction results from tracing over the field degrees of freedom, i.e.,
\begin{equation}\label{Eq: state of the detector}
\rhoh_\textsc{d}=\Tr_\phi\Big( \hat U \rhoh_0 \hat U^\dagger \Big) = \Tr_\phi\Big( \hat U_{N} \cdots \hat U_1 \rhoh_0 \hat U_{1}^\dagger \cdots \hat U_N^\dagger  \Big).
\end{equation}

\subsection{Approximation results}\label{Subsection: Approximation results}

Let us consider a single detector that couples to the field through a switching function $\xi(\tau)$ that is bounded, and continuous except for maybe a finite number of points (in the following, we call this switching function \textit{regular}). For this regular scenario, the joint evolution of the field and the detector is given by the unitary 
\begin{align}\label{Eq: evolution operator continuous}
\hat U_\xi & = \mathcal{T}_\tau \exp\!\bigg(\!\! -\!\ii\int\!\diff\maf z\, \hat h_\text{int}^\xi(\maf z)   \bigg) \\
& = \mathcal{T}_\tau \exp\!\bigg\{\!\! -\!\ii\lambda \int\! \diff\tau\,\xi(\tau) \nonumber \\
& \phantom{==}\times \bigg[\muh_\alpha^\dagger(\tau)\int\!\diff\bm z \sqrt{-g}\,F(\bm z)\,\geno^\alpha(\tau,\bm z) + \text{H.c.}\bigg]\! \bigg\}, \nonumber
\end{align}
where $g = \det g_{\mu\nu}(\tau,\bm z)$ is the determinant of the metric. The final state of the detector is thus
\begin{equation}
\rhoh_\textsc{d}^\xi = \Tr_\phi\Big(\hat U_\xi \rhoh_0 \hat U_\xi^\dagger \Big).
\end{equation}

It turns out that this scenario can be efficiently approximated using a train of delta interactions. Specifically, let $[0, T]$ be the interval where $\xi(\tau)$ has its support, 
where in particular we have $\xi(0)=\xi(T)=0$. The idea is to approximate the phenomenology of the regular switching function with that of a sequence of uniformly spaced Dirac delta pulses. Thus, we can define a train of $N$ sudden interactions associated with $\xi$ and a uniform partition \mbox{$\prt = \{t_j=T j/N,\,j=0,\hdots,N\}$} of the interval $[0,T]$ as
\begin{equation}\label{Eq: chi_xi}
\chi_\xi(\tau; N) = \frac{T}{N} \sum_{j=1}^{N} \xi\bigg(\frac{j-1/2}{N} T \bigg) \delta\bigg(\tau-\frac{j-1/2}{N}T\bigg). \!
\end{equation}
Recalling Eq.~\eqref{Eq: train of deltas}, this amounts to setting 
\begin{equation}
\tau_j = \frac{j-1/2}{N}T,
\end{equation}
and
\begin{equation}
\eta_j = \xi(\tau_j) \,(t_j-t_{j-1}) = \xi(\tau_j) \frac{T}{N}.
\end{equation}
The strength of each sudden interaction is thus determined by both the value of the switching $\xi$ at $t=\tau_j$ and the duration $T/N$ of the interval it is effectively `substituting'. Let us denote with $\rhoh_\textsc{d}(N)$ the final state of the detector according to Eq.~\eqref{Eq: state of the detector} after the train of $N$ sudden interactions $\chi_\xi(\tau;N)$ in Eq.~\eqref{Eq: chi_xi}. 

 Now, let us assume that the detector and the field are initially uncorrelated, so that the initial joint state is of the form
\begin{equation}
\rhoh_0 = \rhoh_{\textsc{d},0}\otimes\rhoh_\phi.
\end{equation}
We can use perturbation theory to write
\begin{equation}\label{Eq: Dyson series delta switching}
\rhoh_\textsc{d}(N) = \rhoh_\textsc{d,0} + \sum_{k=1}^{\infty} \rhoh_\textsc{d}^{(k)}(N)
\end{equation}
for the train of deltas, and
\begin{equation}\label{Eq: Dyson series regular switching}
\rhoh_\textsc{d}^\xi = \rhoh_\textsc{d,0} + \sum_{k=1}^{\infty} \rhoh_\textsc{d}^{(k)\xi}
\end{equation}
for the regular scenario, where $\rhoh_\textsc{d}^{(k)}(N)$ and $\rhoh_\textsc{d}^{(k)\xi}$ are proportional to $\lambda^k$. Then it can be shown that, for each $k\in\mathbb{Z}^+$, under the assumption that the switching function $\xi$, the smearing function $F$, and the initial state of the field $\rhoh_\phi$ are \textit{sufficiently regular}\footnote{To expand on what is meant by sufficiently regular, please see Appendix~\ref{Appendix: Proof of the approximation result}. In short, for most choices of switching and smearing functions, and for the initial states of the field usually employed in the literature (e.g., Hadamard states, both in flat and curved spacetimes), we expect the regularity condition to be satisfied.}, 
\begin{equation}\label{Eq: Approximation result order by order}
\lim_{N\to\infty}\rhoh_\textsc{d}^{(k)}(N) = \rhoh_\textsc{d}^{(k)\xi},
\end{equation}
in the weak operator topology. In fact, the convergence is \textit{at least} as fast as $1/N$. Specifically, for any linear functional $G$,
\begin{equation}\label{Eq: Approximation result order by order bis}
G\big[\rhoh_\textsc{d}^{(k)\xi}\big] = G\big[\rhoh_\textsc{d}^{(k)}(N)\big] + \mathcal{O}(1/N).
\end{equation}
Now, when $\xi$, $F$, and $\rhoh_\phi$ are such that the conditions for Eq.~\eqref{Eq: Approximation result order by order bis} are fulfilled \textit{for all} $k \in \mathbb{Z}^+$, then  the train of deltas approximates the regular scenario \textit{at all orders in perturbation theory}, since
\begin{equation}\label{Eq: Approximation result bis}
G\big[\rhoh_\textsc{d}^\xi\big] = G\big[\rhoh_\textsc{d}(N)\big] + \mathcal{O}(1/N)
\end{equation}
for any linear functional $G$, as long as $\lambda$ is within the radius of convergence of the Dyson series in Eqs.~\eqref{Eq: Dyson series delta switching} and~\eqref{Eq: Dyson series regular switching}. In particular, 
\begin{equation}\label{Eq: Approximation result}
\lim_{N\to\infty} \rhoh_\textsc{d}(N) = \rhoh_\textsc{d}^\xi
\end{equation}
in the weak operator topology. Notice that since it is possible to perform non-perturbative calculations with delta couplings (see, e.g.,~\cite{Simidzija2017nonperturbative,Sahu2022,Erickson2023Gaussian}), Eq.~\eqref{Eq: Approximation result bis} allows to efficiently approximate non-perturbative results, as opposed to the usual perturbative approximations based on truncating the Dyson expansion. While this is still an approximation (that can be made arbitrarily precise by increasing the density of delta couplings), it approximates the full sum of the perturbative expansion, and not only up to a given order. This is why we refer to this technique as a non-perturbative approximation.

\section{Multiple detectors case}\label{Section: Multiple detectors}

In this section we will generalize the framework and results of Sec~\ref{Section: Single detector} to setups involving multiple particle detectors coupling to a quantum field.

\subsection{Setup}\label{Subsection: Setup multiple detectors}

We consider $D$ particle detectors in the general globally hyperbolic spacetime $(\mathcal{M},\maf g)$, whose centres of mass follow trajectories $\maf x_n(\tau_n)$ parametrized by their proper times $\tau_n$, within the regime in which each one of them can be fairly assumed to be Fermi-Walker rigid. 

Each detector is weakly coupled to a quantum field $\phih$. In order to jointly describe the interaction of all the detectors, we will assume that there exists a set of coordinates $(\tau,\bm z)$ of $\mathcal{M}$ such that, for each detector, there is a neighbourhood of its trajectory $\maf x_n(\tau_{n})$ where the coordinates $(\tau,\bm z)$ are a set of Fermi normal coordinates of the detector's centre of mass. In particular, if these neighbourhoods include the supports of each detector's smearing function, then we can write the interaction Hamiltonian weight as  
\begin{equation}\label{Eq: interaction Hamiltonian weight multiple detectors}
\!\!\!\!\hat h_\text{int}(\tau,\bm z) \!=\! \sum_{n=1}^D \!\lambda_n \chi_n(\tau)\!\nobreak\hspace{0.1em} \Big[\!\nobreak\hspace{0.1em} F_n(\bm z) \muh_{n,\alpha}^\dagger(\tau) \geno^\alpha_n(\tau,\bm z)\!\nobreak\hspace{0.05em} +\!\nobreak\hspace{0.05em} \text{H.c.}\!\nobreak\hspace{0.1em} \Big]\!\nobreak\hspace{0.1em}.\!\!
\end{equation}
In complete analogy with Eq.~\eqref{Eq: interaction Hamiltonian}, here $\lambda_n$, $\chi_n$, and $F_n$ are the coupling strength, the switching, and the smearing functions of the $n$-th detector, while $\muh_n^\alpha$ is an arbitrary tensor operator of the $n$-th detector, and $\geno_n^\alpha$ is the generic tensor field operator it couples to. 

In the case of multiple delta couplings, the switching functions can be written as
\begin{equation}\label{Eq: def chi_n,xi_n}
\chi_n(\tau) = \sum_{j=1}^{N_n} \eta_{n,j} \,\delta(\tau-\tau_{n,j}),
\end{equation}
where $N_n$ is the number of times that the $n$-th detector couples to the field, and $\eta_{n,j}$ is the strength of the coupling at time $\tau_{n,j}$. With the interaction Hamiltonian weight given by Eq.~\eqref{Eq: interaction Hamiltonian weight multiple detectors}, the evolution operator in the interaction picture is given by
\begin{align}\label{Eq: complete unitary for multiple detectors}
\!\!\hat U & = \mathcal{T}_\tau\!\exp\!\bigg(\!\!-\ii\!\int\!\diff\bm z\, \hat h_{\text{int}}(\maf z)  \bigg) \\
& = \mathcal{T}_\tau\exp\!\bigg\{ \!\!-\ii\sum_{n=1}^D \lambda_n \sum_{j=1}^{N_n} \eta_{n,j} \nonumber \\
& \phantom{====}\times \bigg[ \muh_{n,j,\alpha}^\dagger\! \int\!\diff\bm z\sqrt{-g_{n,j}}\, F(\bm z) \geno_{n,j}^{\alpha}(\bm z)\! +\! \text{H.c.} \bigg] \!\bigg\}, \nonumber \\
& = \mathcal{T}_\tau \prod_{n,j}\exp\!\bigg\{ \!\!-\ii \lambda_n \eta_{n,j} \nonumber \\
& \phantom{====}\times \bigg[ \muh_{n,j,\alpha}^\dagger\! \int\!\diff\bm z\sqrt{-g_{n,j}} \,F(\bm z) \geno_{n,j}^{\alpha}(\bm z)\! +\! \text{H.c.} \bigg] \!\bigg\}, \nonumber
\end{align} 
where in the last equality, $n$ is running from 1 to $D$, and $j$ runs from 1 to $N_n$, for each $n$. Here, in analogy with the notation used for one detector in Sec.~\ref{Subsection: Setup single detector}, \mbox{$g_{n,j} \equiv \det g_{\mu\nu}(\tau_{n,j},\bm z)$} is the determinant of the metric at $(\tau_{n,j},\bm z)$, and similarly \mbox{$\muh_{n,j,\alpha}^\dagger \equiv \muh_{n,\alpha}^\dagger(\tau_{n,j})$}, and \mbox{$\geno_{n,j}^\alpha(\bm z) \equiv \geno^\alpha(\tau_{n,j},\bm z)$}. Notice that Eq.~\eqref{Eq: complete unitary for multiple detectors} implies that we can write $\hat U$ as a product of unitaries in analogy with Eq.~\eqref{Eq: sequence of unitaries}, where each unitary of the sequence is of the form 
\begin{align}\label{Eq: each unitary multiple detectors}
	\hat U_{n,j} & = \exp\bigg\{\!\!-\ii \lambda_n \eta_{n,j} \\
	& \phantom{===} \times\bigg[  \muh_{n,j,\alpha}^\dagger\!\int\!\diff\bm z \sqrt{-g_{n,j}} \,F(\bm z) \geno^\alpha_{n,j}(\bm z)\!+\!\text{H.c.} \bigg]\! \bigg\}, \nonumber
\end{align}
for some $n,j$. Notice that this amounts to ordering all the $\tau_{n,j}$ and applying the $\hat U_{n,j}$ following the same order. If two times $\tau_{n,j}$ and $\tau_{n',j'}$ coincide, for some $n\neq n'$, then, since we can assume that both detectors do not overlap, the corresponding unitaries $\hat U_{n,j}$ and $\hat U_{n',j'}$ commute, and therefore it does not matter how we decide to order them. This is because if the supports of detectors $n$ and $n'$ at \mbox{$\tau_{n,j}=\tau_{n',j'}$} do not overlap, then the field operators in $\hat U_{n,j}$ and $\hat U_{n',j'}$ have support in spacelike separated regions, and thus commute. Since the detector operators also commute for being associated with different detectors (and then defined over different Hilbert spaces), the exponents of the exponentials that define $\hat U_{n,j}$ and $\hat U_{n',j'}$ commute, and hence so do $\hat U_{n,j}$ and $\hat U_{n',j'}$ themselves.

Finally, after the interaction, the joint state of the $D$ detectors and the field is
\begin{equation}
\rhoh = \hat U \rhoh_0 \hat U^\dagger,
\end{equation}
where $\rhoh_0$ is the initial state of the whole system. In particular, the final joint state of the $D$ detectors after the interaction is obtained by tracing over the field degrees of freedom,
\begin{equation}\label{Eq: state of multiple detectors}
	\rhoh_\textsc{d}=\Tr_\phi\Big( \hat U \rhoh_0 \hat U^\dagger \Big) = \Tr_\phi\Big( \hat U_{N} \cdots \hat U_1 \rhoh_0 \hat U_{1}^\dagger \cdots \hat U_N^\dagger  \Big).\!
\end{equation}

\subsection{Approximation results}\label{Subsection: Approximation results multiple detectors}

Let us consider $D$ particle detectors under the same conditions as in the previous subsection, except for their coupling to the field, which now happens through switching functions $\xi_n(\tau)$ that are bounded, and continuous except for maybe a finite number of points (\textit{regular}, as we called this kind of switching in Sec.~\ref{Subsection: Approximation results}). For this regular scenario, the joint evolution of all the particle detectors and the field is given by the unitary 
\begin{align}
\hat U_\xi & = \mathcal{T}_\tau \exp\!\bigg(\!\! -\!\ii\int\diff\bm z\, \hat h_{\text{int}}^{\xi}(\maf z) \bigg) \\
& = \mathcal{T}_\tau\exp\!\bigg\{ \!\!-\!\ii \sum_{n=1}^D \lambda_n \int\diff\tau \,\xi_n(\tau) \nonumber \\
& \phantom{==} \times \bigg[ \muh_{n,\alpha}^\dagger(\tau)\! \int\!\diff\bm z \sqrt{-g}\, F_n(\bm z) \, \geno_n^\alpha (\tau,\bm z) + \text{H.c.} \bigg]\!\bigg\}, \nonumber
\end{align}
where as before $g=\det g_{\mu\nu}(\tau,\bm z)$ is the determinant of the metric. The final joint state of the $D$ detectors is 
\begin{equation}
	\rhoh_\textsc{d}^\xi = \Tr_\phi\Big( \hat U_\xi \rhoh_0 \hat U_\xi^\dagger \Big).
\end{equation}

As in the single detector case, it turns out that this scenario can be efficiently approximated using a train of delta interactions. Specifically, let us denote with \mbox{$[T_n,T_n+\Delta T_n]$} the interval where $\xi_n(\tau)$ has its support, and in particular \mbox{$\xi_n(T_n)=\xi_n(T_n+\Delta T_n)=0$}. We can then define the train of $N$ sudden interactions associated with $\xi_n$ and a uniform partition \mbox{$\mathcal{P}_n=\{t_{n,j}=T_n +  j \Delta T_n /N,\, j=1,\hdots,N\}$} of the interval $[T_n, T_n+\Delta T_n]$ as
\begin{align}\label{Eq: chi_xi multiple detectors}
	\chi_{\xi_n}(\tau; N) & = \frac{\Delta T_n}{N} \sum_{j=1}^N \xi_n\bigg(T_n + \frac{j-1/2}{N}\Delta T_n  \bigg) \\
	& \phantom{=====} \times \delta\bigg( \tau - T_n - \frac{j-1/2}{N}\Delta_n T \bigg). \nonumber
\end{align} 
From Eq.~\eqref{Eq: def chi_n,xi_n}, this means having $N_n=N$ for all $n$, setting
\begin{equation}
	\tau_{n,j} = T_n + \frac{j-1/2}{N} \Delta T_n,
\end{equation}
and
\begin{equation}
	\eta_{n,j} = \xi_n(\tau_{n,j})(t_{n,j} - t_{n,j-1}) = \xi_n(\tau_{n,j}) \frac{\Delta T_n}{N}.
\end{equation}
This is in complete analogy with the single detector case presented in Sec.~\ref{Subsection: Approximation results}, i.e., the strength of each sudden interaction is determined by both the value of the switching $\xi_n$ at the corresponding time, and the duration of the interval it stands for. 

Now, let us assume that, initially, the detectors and the field are mutually uncorrelated, so that the initial joint state is of the form
\begin{equation}
	\rhoh_0 = \rhoh_{\textsc{d},0} \otimes \rhoh_\phi,
\end{equation}
where $\rhoh_{\textsc{d},0}$ is the initial joint state of the $D$ detectors (which in particular can be correlated). Then, we can use perturbation theory to write
\begin{equation}\label{Eq: series delta multiple}
	\rhoh_\textsc{d}(N) = \sum_{\bm l \in \mathbb{N}_{0}^D} \rhoh_\textsc{d}^{(\bm l)}(N)
\end{equation}
for the train of deltas (where $\mathbb{N}_0 = \mathbb{Z}^+ \cup \{0\}$), and
\begin{equation}\label{Eq: series regular multiple}
	\rhoh_\textsc{d}^\xi = \sum_{\bm l \in \mathbb{N}_0^D}\rhoh_\textsc{d}^{(\bm l)\xi}
\end{equation}
for the regular scenario, where $\rhoh_\textsc{d}^{(\bm l)}(N)$ and $\rhoh_\textsc{d}^{(\bm l)\xi}$ are proportional to
\begin{equation}
\prod_{n=1}^D \lambda_n^{l_n}
\end{equation}
for each $\bm l \in \mathbb{N}_0^D$, and the first term for both series is just the initial state, i.e.,
\begin{equation}
	\rhoh_{\textsc{d}}^{(\bm 0)}(N) = \rhoh_\textsc{d}^{(\bm 0)\xi} = \rhoh_{\textsc{d},0}.
\end{equation}
It can be shown that, for each $\bm l \in \mathbb{N}_0^D$, under the assumption that the switching functions $\xi_n$, the smearing functions $F_n$, and the initial state of the field $\rhoh_\phi$ are \textit{sufficiently regular}\footnote{As in Sec.~\ref{Subsection: Approximation results}, see Appendix~\ref{Appendix: Proof of the approximation result} for details on what is meant by sufficiently regular.},
\begin{equation}\label{Eq: Approximation result order by order multiple}
	\lim_{N\to\infty} \rhoh_\textsc{d}^{(\bm l)}(N) = \rhoh_\textsc{d}^{(\bm l)\xi}
\end{equation}
in the weak operator topology, and the convergence is \textit{at least} as fast as $1/N$, i.e., for any linear functional $G$,
\begin{equation}\label{Eq: Approximation result order by order multiple bis}
	G\big[ \rhoh_\textsc{d}^{(\bm l)\xi} \big] = G \big[ \rhoh_\textsc{d}^{(\bm l)}(N) \big] + \mathcal{O}(1/N).
\end{equation}
In particular, when $\xi_n$, $F_n$, and $\rhoh_\phi$ are such that the conditions for Eq.~\eqref{Eq: Approximation result order by order multiple bis} are fulfilled \textit{for all} $\bm l \in \mathbb{N}_0^D$, then 
\begin{equation}\label{Eq: Approximation result multiple bis}
G\big[ \rhoh_\textsc{d}^\xi \big] = G\big[ \rhoh_\textsc{d}(N)\big] + \mathcal{O}(1/N)
\end{equation}
for any linear functional $G$, as long as $\lambda$ is within the radius of convergence of the Dyson series in Eqs.~\eqref{Eq: series delta multiple} and~\eqref{Eq: series regular multiple}. In particular, the dynamics induced by the train of deltas approximates those of the regular scenario \textit{at all orders in perturbation theory}, i.e.,
\begin{equation}\label{Eq: Approximation result multiple}
	\lim_{N\to\infty} \rhoh_\textsc{d}(N) = \rhoh_\textsc{d}^\xi
\end{equation}
in the weak operator topology. This means that the same conclusions that we obtained for the single detector scenario---and in particular the possibility of efficiently approximating non-perturbative results---can be extended to setups involving multiple detectors as well. Notice that the density matrix that we get to approximate describes the joint state of all the detectors, and not just one of them. This means that we can use this technique to extract conclusions regarding not only the state of each individual detector, but also the correlations between them.

\section{Examples}\label{Section: Examples}

After introducing the approximation method in Secs.~\ref{Section: Single detector} and~\ref{Section: Multiple detectors}, here we will use it in two specific scenarios, involving, respectively, one and two detectors. For each case, we will also analyze the rate of convergence of the approximate density matrices towards the exact ones, showing explicitly that the approximations are at least as efficient as guaranteed by the results stated in Secs.~\ref{Subsection: Approximation results} and~\ref{Subsection: Approximation results multiple detectors} (and proved in Appendix~\ref{Appendix: Proof of the approximation result}).

\subsection{Single detector example}\label{Subsection: Single detector example}

In this example, we consider a two-level Unruh-DeWitt particle detector~\cite{Unruh,DeWitt} at rest in a ($1+3$)-dimensional Minkowski spacetime, with ground and excited states $\ket{g}$ and $\ket{e}$, separated by an energy gap $\Omega$, and a Gaussian shape,
\begin{equation}
F(\bm x) = \frac{1}{\sqrt{\pi^3} \sigma^3} \, e^{-\bm x^2/\sigma^2}.
\end{equation}
The detector is coupled to a massless scalar quantum field via an interaction Hamiltonian
\begin{equation}
\hat H_{\text{int}}(t)=\lambda \,\xi(t) \muh(t) \!\int\!\diff\bm x\, F(\bm x) \,\phih(t,\bm x),
\end{equation}
where $\lambda$ is the coupling strength as in Eq.~\eqref{Eq: interaction Hamiltonian}, and $\muh$ is the monopole moment of the detector in the interaction picture,
\begin{equation}
\muh(t) = \proj{g}{e}\,e^{-\ii \Omega t} + \proj{e}{g}\,e^{\ii\Omega t}.
\end{equation}
The field operator $\phih$ can be expanded in plane-wave modes as
\begin{equation}\label{Eq: scalar field plane waves expansion}
\hat\phi (t,\bm x)= \int \frac{\diff \bm k}{\sqrt{(2\pi)^3 2|\bm k|}} \, \big(\hat a_{\bm k} e^{-\ii(|\bm k|t - \bm k \cdot \bm x)} + \text{H. c.}\big).
\end{equation} 
Finally, $\xi$ is a switching function supported in $[0,T]$ that is bounded and continuous except for maybe a finite number of points, as in Sec.~\ref{Subsection: Approximation results}. 

We take the detector and the field to be initially uncorrelated and in their respective ground states,
\begin{equation}
\rhoh_0 = \rhoh_{\textsc{d},0} \otimes \proj{0}{0}, \qquad \rhoh_{\textsc{d},0} = \proj{g}{g}.
\end{equation}
Using perturbation theory, we can write the state of the detector at the end of the interaction process as
\begin{equation}
\rhoh_\textsc{d} = \rhoh_{\textsc{d},0} + \rhoh^{(2)}_{\textsc{d}} + \mathcal{O}(\lambda^4),
\end{equation}
where the terms of odd order in the perturbative expansion are zero because the vacuum $\ket{0}$ is a zero-mean Gaussian state, and therefore its odd-point functions are zero. Thus, the final state of the detector is diagonal in the $\{\ket{g},\ket{e}\}$ basis, and can be completely characterized by the excitation probability,
\begin{equation}
P_{\text{e}} = \bra{e} \rhoh_\textsc{d} \ket{e} = \bra{e} \rhoh^{(2)}_{\textsc{d}} \ket{e} + \mathcal{O}(\lambda^4).
\end{equation}
For a massless scalar quantum field, up to second order in $\lambda$ (see, e.g.,~\cite{Erickson2021click}),
\begin{equation}
P_{\text{e}}=\lambda^2 \int \frac{\diff\bm k}{ (2\pi)^3 2 |\bm k|} \,|\tilde{\xi}(|\bm k|+\Omega)|^2 |\tilde{F}(\bm k)|^2, 
\end{equation}
where $\tilde{\xi}$ and $\tilde{F}$ are the Fourier transforms\footnote{The convention we followed here for the definition of the Fourier transform is that
\begin{equation}
\tilde{G}(\bm v) \coloneqq \int \diff \bm u \, G(\bm u) \, e^{-\ii \bm u \cdot \bm v}.
\end{equation}
} of $\xi$ and $F$, respectively. For the Gaussian shape specified before, which has spherical symmetry, we get
\begin{equation}\label{Eq: excitation probability example}
P_{\text{e}} = \frac{\lambda^2}{4\pi^2} \int_0^\infty \diff |\bm k| \, |\bm k| e^{-|\bm k|^2\sigma^2/2} \,|\tilde{\xi}(|\bm k|+\Omega)|^2.
\end{equation}

As a particular application of the result stated in Sec.~\ref{Subsection: Approximation results}, Eq.~\eqref{Eq: excitation probability example} can be approximated by the transition probability for a scenario where the detector couples to the field through the train of $N$ sudden interactions given by
\begin{equation}
\chi_\xi(t; N) = \frac{T}{N} \sum_{j=1}^{N} \xi\bigg(\frac{j-1/2}{N}T\bigg) \delta\bigg(t-\frac{j-1/2}{N}T\bigg).\!
\end{equation}
In this case, the excitation probability is given (up to second order in $\lambda$) by
\begin{align}\label{Eq: excitation probability approximate}
P_{\text{e}}(N) &  = \frac{\lambda^2 T^2}{N^2}\!\! \sum_{j,j'=1}^N \!\xi\bigg(\frac{j-1/2}{N}T\bigg)\xi\bigg(\frac{j'-1/2}{N}T\bigg) \!\\
& \!\times \frac{ e^{\ii \Omega T (j-j')/N}}{4\pi^2} \!\!\int_0^\infty\!\!\! \diff |\bm k| |\bm k| e^{-|\bm k|^2\sigma^2/2} e^{\ii |\bm k| T (j-j')/N}. \nonumber
\end{align}
In the following we will see, for specific switching functions $\xi$, how Eq.~\eqref{Eq: excitation probability approximate} approximates Eq.~\eqref{Eq: excitation probability example} increasingly well as $N$ increases, and does so with the predicted efficiency.

\subsubsection{Heaviside switching}

The first example that we showcase is a switching function that is constant in its support,
\begin{equation}\label{Eq: Heaviside switching}
\xi(t) = \textsc{I}_{(0,T)}(t),
\end{equation}
where $\textsc{I}_{(0,T)}$ is the indicator function that equals 1 in the interval $(0,T)$, and zero everywhere else. In this case,
\begin{equation}\label{Eq: Fourier transform of Heaviside switching}
\tilde{\xi}(k) = \frac{2 e^{-\ii k T /2}}{k} \sin\!\bigg( \frac{k T}{2} \bigg).
\end{equation}
From Eq.~\eqref{Eq: excitation probability example}, we have
\begin{align}
P_{\text{e}} & = \frac{\lambda^2}{\pi^2} \int_0^\infty \diff |\bm k| \, \frac{|\bm k|\, e^{-|\bm k|^2 \sigma^2/2}}{(|\bm k| + \Omega)^2} \sin^2\!\bigg[ \frac{(|\bm k|+\Omega) T}{2} \bigg] \\
& = \frac{\lambda^2}{\pi^2} \int_0^\infty \diff\kappa\, \frac{\kappa e^{-\kappa^2 s^2 /2}}{(\kappa + \gamma)^2} \sin^2\!\bigg( \frac{\kappa + \gamma}{2} \bigg), \label{Eq: excitation probability Heaviside}
\end{align}
where the second line explicitly shows that the result only depends on the dimensionless parameters \mbox{$\gamma = \Omega T$}, and \mbox{$s = \sigma/T$}, and we also denoted \mbox{$\kappa = |\bm k| T$}. Similarly, from Eq.~\eqref{Eq: excitation probability approximate} we get the excitation probability associated with the train of sudden interactions $\chi_\xi$ that approximates the Heaviside switching,
\begin{align}\label{Eq: excitation probability Heaviside approximate}
P_{\text{e}}(N)= \frac{\lambda^2}{N^2} \!\!\! \sum_{j,j'=1}^{N}\!\!\! \frac{e^{\ii\gamma(j-j')/N}}{4\pi^2}\!\! \int_0^\infty\!\!\!\diff\kappa\, \kappa e^{-\kappa^2 s^2/2} e^{\ii \kappa (j-j')/N}.
\end{align}

In Fig.~\ref{Fig: 1d Heaviside}a we see how $P_{\text{e}}(N)$ as given in Eq.~\eqref{Eq: excitation probability Heaviside approximate} converges to the exact value $P_{\text{e}}$ given in Eq.~\eqref{Eq: excitation probability Heaviside}, for the case \mbox{$\gamma=s=1$}. Moreover, Fig.~\ref{Fig: 1d Heaviside}b shows that the rate of convergence of the approximation is as fast as $1/N^2$, which is better than the minimum rate guaranteed by the approximation result of Sec.~\ref{Subsection: Approximation results}.

\begin{figure*}
\begin{center}
\includegraphics[scale=0.82]{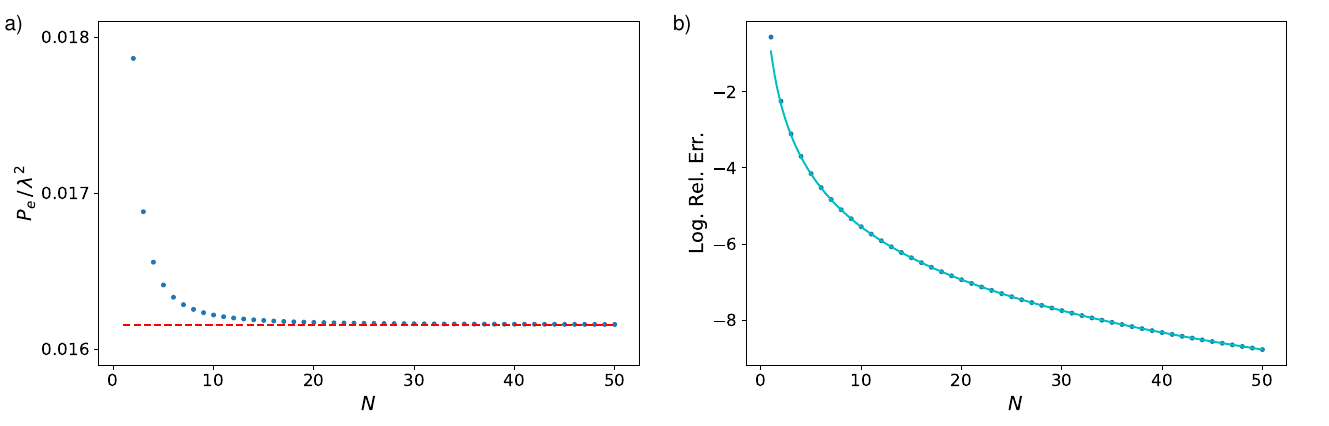}
\caption{a) Excitation probability of a two-level UDW detector coupled to a quantum scalar field through a train of delta couplings approximating a Heaviside switching, as a function of the number of couplings $N$. Here we set the relevant parameters $\gamma$ and $s$ to 1. The constant red dashed line marks the value of the excitation probability for the exact Heaviside switching. b) Logarithmic relative error of the excitation probability for the Heaviside switching approximated by delta couplings, as a function of the number of couplings $N$. The function represented by the solid cyan line is proportional to $1/N^2$, marking the rate of convergence of the approximation.}
\label{Fig: 1d Heaviside}
\end{center}
\end{figure*}

\subsubsection{Truncated Gaussian switching}

The second example that we consider is a Gaussian switching where the Gaussian tails are cut beyond $q$ times its variance:
\begin{equation}
\xi(t) = \exp\bigg[ -\bigg( \frac{t}{qT} - \frac{1}{2q} \bigg)^{\!\!2}\, \bigg]\, \textsc{I}_{(0,T)}(t).
\end{equation}
Notice that the Gaussian reaches its maximum at \mbox{$t=T/2$}, i.e., the centre of the switching's compact support. Its Fourier transform reads
\begin{equation}
\!\!\!\tilde{\xi}(k)\! =\! \sqrt{\pi}q T e^{-\ii k T /2} e^{- k^2 q^2 T^2/ 4}\! \Re\!\bigg[\! \erf\!\bigg(\frac{1}{2q}+\frac{\ii k q T}{2} \bigg) \! \bigg].\!
\end{equation}
From Eq.~\eqref{Eq: excitation probability example}, we have
\begin{align}
P_{\text{e}} & = \frac{\lambda^2 q^2 T^2}{4 \pi} \int_0^\infty \diff |\bm k| \, |\bm k| \, e^{-|\bm k|^2\sigma^2/2}\,e^{-(|\bm k|+\Omega)^2q^2T^2/2} \nonumber \\
& \phantom{======} \times \Re\bigg\{ \erf\bigg[ \frac{1}{2 q} + \frac{\ii (|\bm k|+\Omega) qT}{2} \bigg] \bigg\}^2 \\
& = \frac{\lambda^2 q^2}{4\pi} \int_0^\infty \diff\kappa\, \kappa\, e^{-\kappa^2 s^2/2}\,e^{-(\kappa+\gamma)^2q^2/2} \nonumber\\
& \phantom{======} \times \Re\bigg\{ \erf\bigg[ \frac{1}{2q} + \frac{\ii (\kappa+\gamma)q}{2} \bigg] \bigg\}^2, \label{Eq: excitation probability Gaussian}
\end{align}
where, as before, we have shown the explicit dependence on the dimensionless parameters \mbox{$\gamma = \Omega T$}, and \mbox{$s = \sigma/T$}, and the integration is performed with respect to the dimensionless variable \mbox{$\kappa = |\bm k| T$}. Again, from Eq.~\eqref{Eq: excitation probability approximate} we get the excitation probability associated with the train of delta couplings $\chi_\xi$ that yields a transition probability that approximates that of the truncated Gaussian switching,
\begin{align}\label{Eq: excitation probability Gaussian approximate}
\!\!P_{\text{e}}(N) & = \frac{\lambda^2}{N^2} \!\!\! \sum_{j,j'=1}^{N}\!\!\! e^{-(2j-N-1)^2\!/4q^2N^2} e^{-(2j'-N-1)^2\!/4q^2N^2} \nonumber \\
&\phantom{\;} \times \frac{e^{\ii\gamma(j-j')/N}}{4\pi^2}\!\! \int_0^\infty\!\!\!\diff\kappa\, \kappa e^{-\kappa^2 s^2/2} e^{\ii \kappa (j-j')/N}.
\end{align}

In Fig.~\ref{Fig: 1d Gaussian}a we see how $P_{\text{e}}(N)$---as given in Eq.~\eqref{Eq: excitation probability Gaussian approximate}---converges to the value of $P_{\text{e}}$ given in Eq.~\eqref{Eq: excitation probability Gaussian}, for the case \mbox{$\gamma=s=q=1$}. Moreover, Fig.~\ref{Fig: 1d Gaussian}b shows that the rate of convergence of the approximation is once more as fast as $1/N^2$, which is again better than the minimum rate guaranteed by the approximation result of Sec.~\ref{Subsection: Approximation results}.

\begin{figure*}
\begin{center}
\includegraphics[scale=0.82]{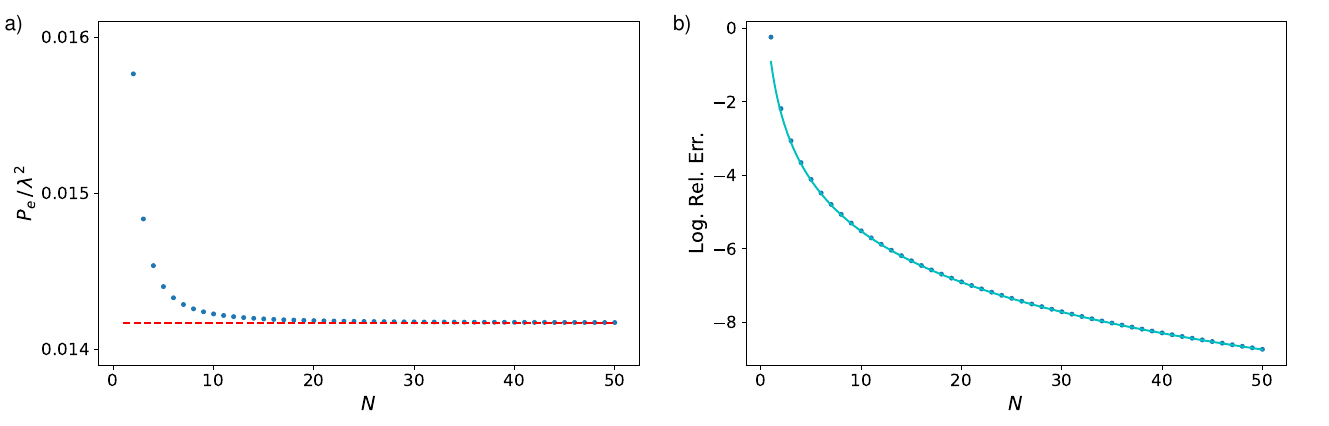}
\caption{a) Excitation probability of a two-level UDW detector coupled to a quantum scalar field through a train of delta couplings approximating a Gaussian switching truncated at $q=1$ variances, as a function of the number of couplings $N$. Here we set the relevant parameters $\gamma$ and $s$ to 1. The constant red dashed line marks the value of the excitation probability for the exact truncated Gaussian switching. b) Logarithmic relative error of the excitation probability for the Gaussian switching approximated by delta couplings, as a function of the number of couplings $N$. The function represented by the solid cyan line is proportional to $1/N^2$, marking the rate of convergence of the approximation.}
\label{Fig: 1d Gaussian}
\end{center}
\end{figure*}


\subsubsection{Smooth bump switching}

As a third example, let us consider the switching to be the $C^\infty$  bump function given by
\begin{equation}
\xi(t) = \begin{dcases*} 
\exp\bigg[ -\frac{T^2}{4t(T-t)} \bigg] & if $0< t < T$,\\
\,0 & otherwise,\\
\end{dcases*}
\end{equation}
which is compactly supported in $[0,T]$, infinitely differentiable on $\R{}$, and reaches its maximum at $t=T/2$. Its Fourier transform does not admit a closed form. However, let us define
\begin{equation}
\beta(t) = \begin{dcases*} 
\exp\bigg[ -\frac{1}{4t(1-t)} \bigg] & if $0 < t < 1$,\\
\,0 & otherwise,\\
\end{dcases*}
\end{equation}
which does not depend on $T$. This function is related to $\xi$ by $\xi(t)=\beta(t/T)$. Thus, their Fourier transforms satisfy
\begin{equation}
\tilde{\xi}(k) = T \tilde{\beta}(k T),
\end{equation}
and hence, from Eq.~\eqref{Eq: excitation probability example}, we can write
\begin{equation}\label{Eq: excitation probability bump}
P_{\text{e}} = \frac{\lambda^2}{4 \pi^2} \int_0^{\infty} \diff\kappa\,\kappa \,e^{-\kappa^2 s^2/2} |\tilde{\beta}(\kappa+\gamma)|^2,
\end{equation}
which shows once again the complete dependence of $P_{\text{e}}$ on the dimensionless parameters $\gamma$ and $s$. Meanwhile, from Eq.~\eqref{Eq: excitation probability approximate} we get the excitation probability associated with the train of delta couplings $\chi_\xi$ used to approximate that of the bump switching,
\begin{align}\label{Eq: excitation probability bump approximate}
P_{\text{e}}(N) & = \frac{\lambda^2}{N^2}  \sum_{j,j'=1}^{N} \exp\bigg[-\frac{N^2}{4(j-1/2)(N-j+1/2)}\bigg] \nonumber \\
& \phantom{=====\!} \times \exp\bigg[-\frac{N^2}{4(j'-1/2)(N-j'+1/2)}\bigg]  \nonumber \\
&\phantom{\;} \times \frac{e^{\ii\gamma(j-j')/N}}{4\pi^2}\!\! \int_0^\infty\!\!\!\diff\kappa\, \kappa e^{-\kappa^2 s^2/2} e^{\ii \kappa (j-j')/N}.
\end{align}

In Fig.~\ref{Fig: 1d Bump}a we see how, for the case \mbox{$\gamma=s=1$}, the approximated values $P_{\text{e}}(N)$ (as given in Eq.~\eqref{Eq: excitation probability bump approximate}) converge to the exact value of $P_{\text{e}}$ (given in Eq.~\eqref{Eq: excitation probability bump}). Moreover, Fig.~\ref{Fig: 1d Bump}b shows that the error incurred by the approximation is upper bounded by a function that decays faster than $1/N^5$,  much faster than the minimum rate guaranteed by the  result of Sec.~\ref{Subsection: Approximation results}. In fact, the increase of the rate of convergence of the approximation in this example with respect to the two previous ones is not surprising, since the approximation results in~\cite{Chui1971} suggest that switching functions of higher differentiability class should admit tighter convergence bounds.

\begin{figure*}
\begin{center}
\includegraphics[scale=0.82]{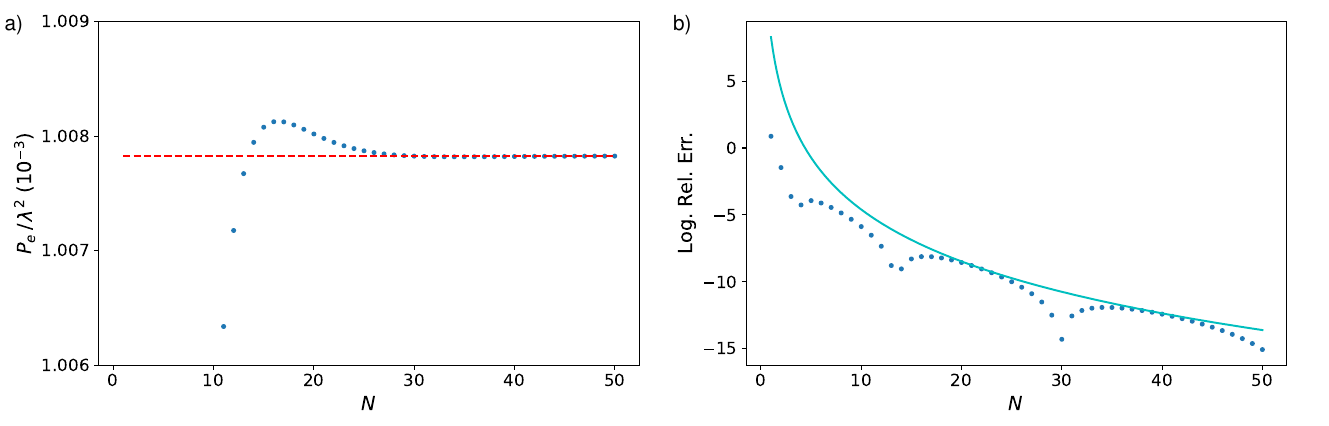}
\caption{a) Excitation probability of a two-level UDW detector coupled to a quantum scalar field through a train of delta couplings approximating a bump switching, as a function of the number of couplings $N$. Here we set the relevant parameters $\gamma$ and $s$ to 1. The constant red dashed line marks the value of the excitation probability for the exact bump switching. b) Logarithmic relative error of the excitation probability for the bump switching approximated by delta couplings, as a function of the number of couplings $N$. The function represented by the solid cyan line is proportional to $1/N^{5.6}$, upper bounding the rate of convergence of the approximation.}
\label{Fig: 1d Bump}
\end{center}
\end{figure*}

\subsection{Two detectors example}\label{Subsection: Two detectors example}

The setup considered for this example consists of two identical two-level Unruh-DeWitt detectors, A and B, both at rest at positions $\bm x_\textsc{a}$ and $\bm x_{\textsc{b}}$ in a \mbox{$(1+3)$-di}mensional Minkowski spacetime. We consider the ground and excited states of both detectors to be separated by the same energy gap $\Omega$, and their shapes to be (normalized) hard spheres of radius $R$,
\begin{equation}
F_i(\bm x) = \frac{3}{4 \pi R^3}\,\theta(R-|\bm x - \bm x_i|), \quad i=\textsc{A},\textsc{B}.
\end{equation}
As in the single detector example, here we consider a linear coupling between the detectors and a massless scalar quantum field, given in the interaction picture by the interaction Hamiltonian
\begin{align}
\hat H_{\text{int}}(t) & = \lambda \xi(t) \bigg[ \muh_\textsc{a}(t) \int\diff\bm x \, F_\textsc{a}(\bm x)\, \phih(t,\bm x) \nonumber \\ 
&\phantom{======} + \muh_\textsc{b}(t) \int\diff\bm x\, F_{\textsc{b}}(\bm x) \,\phih(t,\bm x)\bigg],
\end{align}
where $\lambda$ is the coupling strength as in Eq.~\eqref{Eq: interaction Hamiltonian}, $\muh_i$ is the monopole moment of detector $i$, which can be written in the interaction picture as
\begin{equation}
\muh_i(t) = \proj{g_i}{e_i}\,e^{-\ii \Omega t} + \proj{e_i}{g_i}\,e^{\ii\Omega t},
\end{equation}
the field operator $\phih$ can be expanded in plane-wave modes as in Eq.~\eqref{Eq: scalar field plane waves expansion}, and finally $\xi$ is again a switching function supported in $[0,T]$ that is bounded and continuous except for maybe a finite number of points, as in Sec.~\ref{Subsection: Approximation results multiple detectors}. Notice that, for simplicity, in this example we have chosen the coupling strengths and the switching functions to be the same for both detectors, although this is of course not necessary in the most general case. 

Initially, we consider the detectors and the field to be uncorrelated and in their ground states,
\begin{equation}
\rhoh_0 = \proj{g_\textsc{a}}{g_\textsc{a}} \otimes \proj{g_{\textsc{b}}}{g_{\textsc{b}}} \otimes \proj{0}{0}.
\end{equation}
Perturbation theory allows us to write the joint state of the two detectors after the interaction as
\begin{equation}
\rhoh_\textsc{ab} = \rhoh_{\textsc{ab},0} + \rhoh^{(2)}_{\textsc{ab}} + \mathcal{O}(\lambda^4),
\end{equation}
where $\rhoh_{\textsc{ab},0}=\proj{g_\textsc{a}}{g_\textsc{a}} \otimes \proj{g_\textsc{b}}{g_\textsc{b}}$ and, as in the single detector case, the terms of odd order are zero because the odd-point functions of the vacuum state of the field are zero. In the ordered basis \mbox{$\{\ket{g_\textsc{a}g_{\textsc{b}}},\ket{g_{\textsc{a}}e_{\textsc{b}}},\ket{e_{\textsc{a}}g_\textsc{b}},\ket{e_\textsc{a}e_\textsc{b}}\}$}, the second order contribution can be written in matrix form as (see, e.g.,~\cite{Pozas2015,NegativityCovariance})
\begin{equation}\label{Eq: rho 2 two detectors}
\rhoh_{\textsc{ab}}^{(2)} = \begin{pmatrix}
-\mathcal{L}_{\textsc{aa}} - \mathcal{L}_\textsc{bb} & 0 & 0 & \mathcal{M}^* \\
0 & \mathcal{L}_\textsc{bb} & \mathcal{L}_\textsc{ab}^* & 0 \\
0 & \mathcal{L}_\textsc{ab} & \mathcal{L}_\textsc{aa} & 0 \\
\mathcal{M} & 0 & 0 & 0 \\
\end{pmatrix}.
\end{equation}
Now, let us restrict ourselves to the case in which the interaction of the two detectors happens in regions of spacetime that are spacelike separated, i.e., when
\begin{equation}
T < D - 2R,
\end{equation}
where $D = |\bm x_\textsc{a} - \bm x_\textsc{b}|$ denotes the spatial distance between the two detectors. Then, we get the following expressions for the terms in Eq.~\eqref{Eq: rho 2 two detectors}:
\begin{align}
& \mathcal{L}_{\textsc{aa}} = \mathcal{L}_\textsc{bb} = \frac{\lambda^2}{4 \pi^2} \!\int_0^\infty \!\!\!\diff |\bm k| \, |\bm k| \,|\tilde{\xi}(|\bm k|+\Omega)|^2 |\tilde{F}(|\bm k|)|^2,  \label{Eq: Lii}\\
& \mathcal{L}_{\textsc{ab}} = \!\frac{\lambda^2}{4\pi^2 D}\!\!\int_0^\infty\!\!\!\! \diff |\bm k| \sin(|\bm k| D) |\tilde{\xi}(|\bm k|+\Omega)|^2 |\tilde{F}(|\bm k|)|^2, \!\! \label{Eq: Lab}\\
& \mathcal{M} = -\frac{\lambda^2}{4 \pi^2 D} \int_0^\infty \diff |\bm k| \sin(|\bm k|D) \tilde{\xi}(|\bm k|-\Omega) \nonumber \\
& \phantom{==============} \times \tilde{\xi}(|\bm k| + \Omega)^* \, |\tilde{F}(|\bm k|)|^2, \label{Eq: M}
\end{align}
where $\tilde{\xi}$ and $\tilde{F}$ are the Fourier transforms of the switching function $\xi$ and 
\begin{equation}
F(\bm x) = \frac{3}{4\pi R^3}\,\theta(R-|\bm x|),
\end{equation}
respectively. Notice that to arrive at Eqs.~\eqref{Eq: Lii}--\eqref{Eq: M} we used that
\begin{equation}
F(\bm x) = F_\textsc{a}(\bm x + \bm x_\textsc{a}) = F_{\textsc{b}}(\bm x+\bm x_\textsc{b}),
\end{equation}
as well as the spherical symmetry of $F$, which implies
\begin{equation}
\tilde{F}(\bm k) = e^{\ii \bm k \cdot \bm x_\textsc{a}} \tilde{F}_\textsc{a}(\bm k)= e^{\ii \bm k \cdot \bm x_\textsc{b}} \tilde{F}_\textsc{b}(\bm k) = \tilde{F}(|\bm k|),
\end{equation}
where
\begin{equation}
\tilde{F}(|\bm k|) = \frac{3}{(|\bm k|R)^3} \big[ \sin(|\bm k|R) - |\bm k| R\cos(|\bm k|R) \big].
\end{equation}
Finally, we choose the switching function for this example to be of the Heaviside type, 
\begin{equation}
\xi(t) = \textsc{I}_{(0,T)}(t),
\end{equation}
with Fourier transform
\begin{equation}
\tilde{\xi}(k) = \frac{2 e^{-\ii k T /2}}{k} \sin\!\bigg( \frac{k T}{2} \bigg).
\end{equation}
Eqs.~\eqref{Eq: Lii}--\eqref{Eq: M} can then be written in terms of dimensionless parameters as
\begin{align}
& \mathcal{L}_{\textsc{aa}} = \mathcal{L}_\textsc{bb} = \frac{\lambda^2}{\pi^2} \int_0^\infty \!\!\diff \kappa \, \frac{\sin^2\big[ (\kappa+\gamma)/2 \big]}{\kappa} |\tilde{F}(\kappa)|^2,  \label{Eq: Lii spec}\\
& \mathcal{L}_{\textsc{ab}} = \frac{\lambda^2}{\pi^2 d} \int_0^\infty \!\!\diff \kappa \, \sin(\kappa d) \, \frac{\sin^2\big[ (\kappa+\gamma)/2 \big]}{\kappa^2} |\tilde{F}(\kappa)|^2,  \label{Eq: Lab spec}\\
& \mathcal{M} = -\frac{\lambda^2 e^{\ii\gamma}}{2 \pi^2 d} \int_0^\infty \diff \kappa \, \frac{\sin(\kappa d)(\cos\gamma - \cos\kappa)}{\kappa^2 - \gamma^2} \, |\tilde{F}(\kappa)|^2, \label{Eq: M spec} 
\end{align}
where $\gamma=\Omega T$, $d=D/T$, $\kappa =  |\bm k| T$, and
\begin{equation}
\tilde{F}(\kappa) = \frac{3}{(\kappa r)^3} \big[ \sin(\kappa r) - \kappa r \cos(\kappa r) \big],
\end{equation}
with $r=R/T$. 

As a particular application of the approximation result of Sec.~\ref{Subsection: Approximation results multiple detectors}, we know that Eqs.~\eqref{Eq: Lii spec}--\eqref{Eq: M spec} can be approximated by their counterparts in the scenario where the detectors couple to the field through the train of sudden interactions given by
\begin{equation}
\chi_\xi(t; N) = \frac{T}{N} \sum_{j=1}^{N} \delta\bigg(t-\frac{j-1/2}{N}T\bigg).
\end{equation}
The terms that approximate Eqs.~\eqref{Eq: Lii spec}--\eqref{Eq: M spec} can then be written as
\begin{align}
& \!\!\!\mathcal{L}_\textsc{aa}(N) = \mathcal{L}_\textsc{bb}(N) = \frac{\lambda^2}{N^2}\!\! \sum_{j,j'=1}^N \frac{1}{4 \pi^2} \int_{0}^\infty \diff \kappa\, \kappa\, |\tilde{F}(\kappa)|^2 \nonumber \\
& \phantom{==============\,}\times \cos\bigg[\frac{(\kappa+\gamma)(j-j')}{N}\bigg] ,\label{Eq: Lii approximated}\\
& \!\!\!\mathcal{L}_\textsc{ab}(N) = \frac{\lambda^2}{N^2}\!\! \sum_{j,j'=1}^N \frac{1}{4\pi^2 d} \int_0^\infty\diff \kappa\,\sin(\kappa d) \, |\tilde{F}(\kappa)|^2 \nonumber \\
& \phantom{==============\,}\times \cos\bigg[\frac{(\kappa+\gamma)(j-j')}{N}\bigg] ,\label{Eq: Lab approximated} \\
& \!\!\!\mathcal{M}(N) = - \frac{\lambda^2}{N^2} \!\!\sum_{j,j'=1}^N \frac{e^{\ii\gamma(j+j')/T}}{4 \pi^2 d} \int_0^\infty \diff \kappa\, \sin(\kappa d) \,|\tilde{F}(\kappa)|^2 \nonumber \\
& \phantom{=================\,\,}\times \cos\bigg[\frac{\kappa (j-j')}{N}\bigg]. \label{Eq: M approximated}
\end{align}

In Figs.~\ref{Fig: 2d Lii},~\ref{Fig: 2d Lab}, and~\ref{Fig: 2d M}, we see how, for the case \mbox{$\gamma=1$}, \mbox{$d=1.2$}, and \mbox{$r=0.1$}, the approximated values given in Eqs.~\eqref{Eq: Lii approximated}--\eqref{Eq: M approximated} converge to the exact values given in Eqs.~\eqref{Eq: Lii spec}--\eqref{Eq: M spec}. In particular, Figs.~\ref{Fig: 2d Lii}b,~\ref{Fig: 2d Lab}b, and~\ref{Fig: 2d M}d show that the errors committed by the approximations are upper bounded by functions decaying at least as fast as $1/N$, in agreement with the approximation result of Sec.~\ref{Subsection: Approximation results multiple detectors}.

\begin{figure*}
\begin{center}
\includegraphics[scale=0.82]{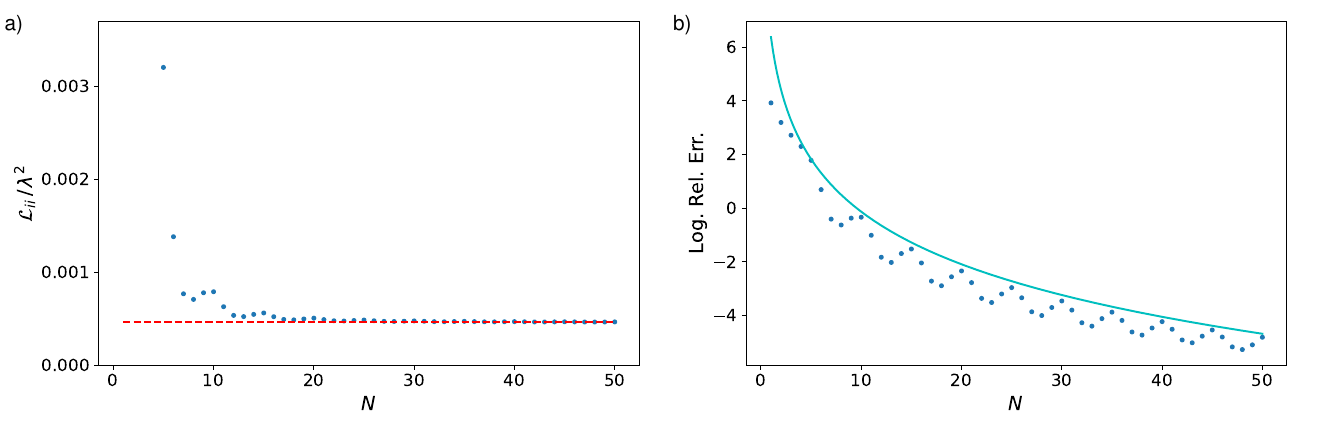}
\caption{a) $\mathcal{L}_{\textsc{aa}}=\mathcal{L}_{\textsc{bb}} \equiv\mathcal{L}_{ii}$ term for a pair of identical two-level UDW detectors coupled to a quantum scalar field in spacelike separated regions through a train of delta couplings approximating a Heaviside switching, as a function of the number of couplings $N$. Here we set the relevant parameters to $\gamma=1$, $d=1.2$, and $r=0.1$. The constant red dashed line marks the value of the $\mathcal{L}_{ii}$ term for the exact Heaviside switching. b) Logarithmic relative error of the $\mathcal{L}_{ii}$ term for the Heaviside switching approximated by delta couplings, as a function of the number of couplings $N$. The function represented by the solid cyan line is proportional to $1/N^{2.8}$, upper bounding the rate of convergence of the approximation.}
\label{Fig: 2d Lii}
\end{center}
\end{figure*}

\begin{figure*}
\begin{center}
\includegraphics[scale=0.82]{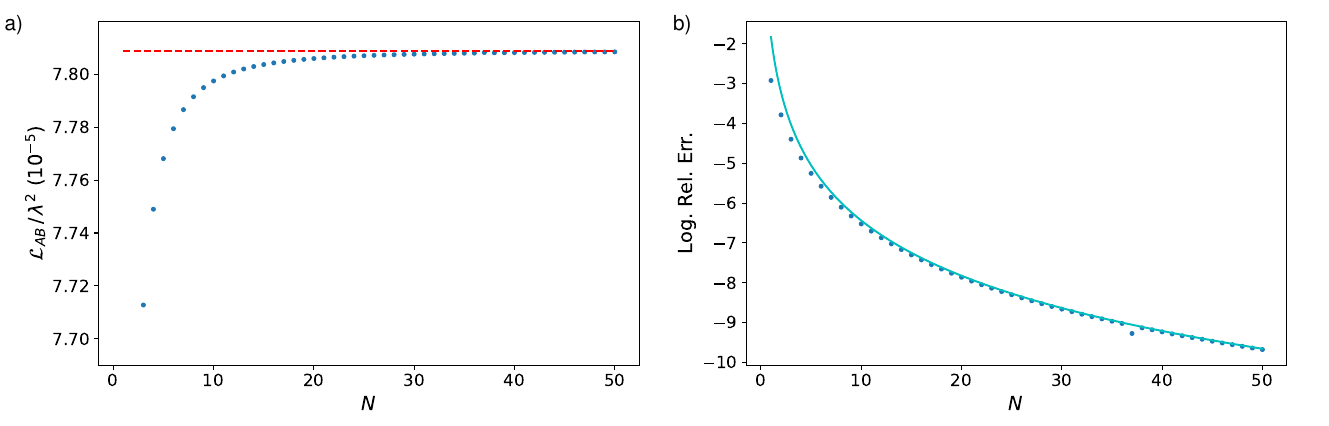}
\caption{a) $\mathcal{L}_\textsc{ab}$ term for a pair of identical two-level UDW detectors coupled to a quantum scalar field in spacelike separated regions through a train of delta couplings approximating a Heaviside switching, as a function of the number of couplings $N$. Here we set the relevant parameters to $\gamma=1$, $d=1.2$, and $r=0.1$. The constant red dashed line marks the value of the $\mathcal{L}_\textsc{ab}$ term for the exact Heaviside switching. b) Logarithmic relative error of the $\mathcal{L}_\textsc{ab}$ term for the Heaviside switching approximated by delta couplings, as a function of the number of couplings $N$. The function represented by the solid cyan line is proportional to $1/N^{2}$, marking the rate of convergence of the approximation.}
\label{Fig: 2d Lab}
\end{center}
\end{figure*}

\begin{figure*}
\begin{center}
\includegraphics[scale=0.82]{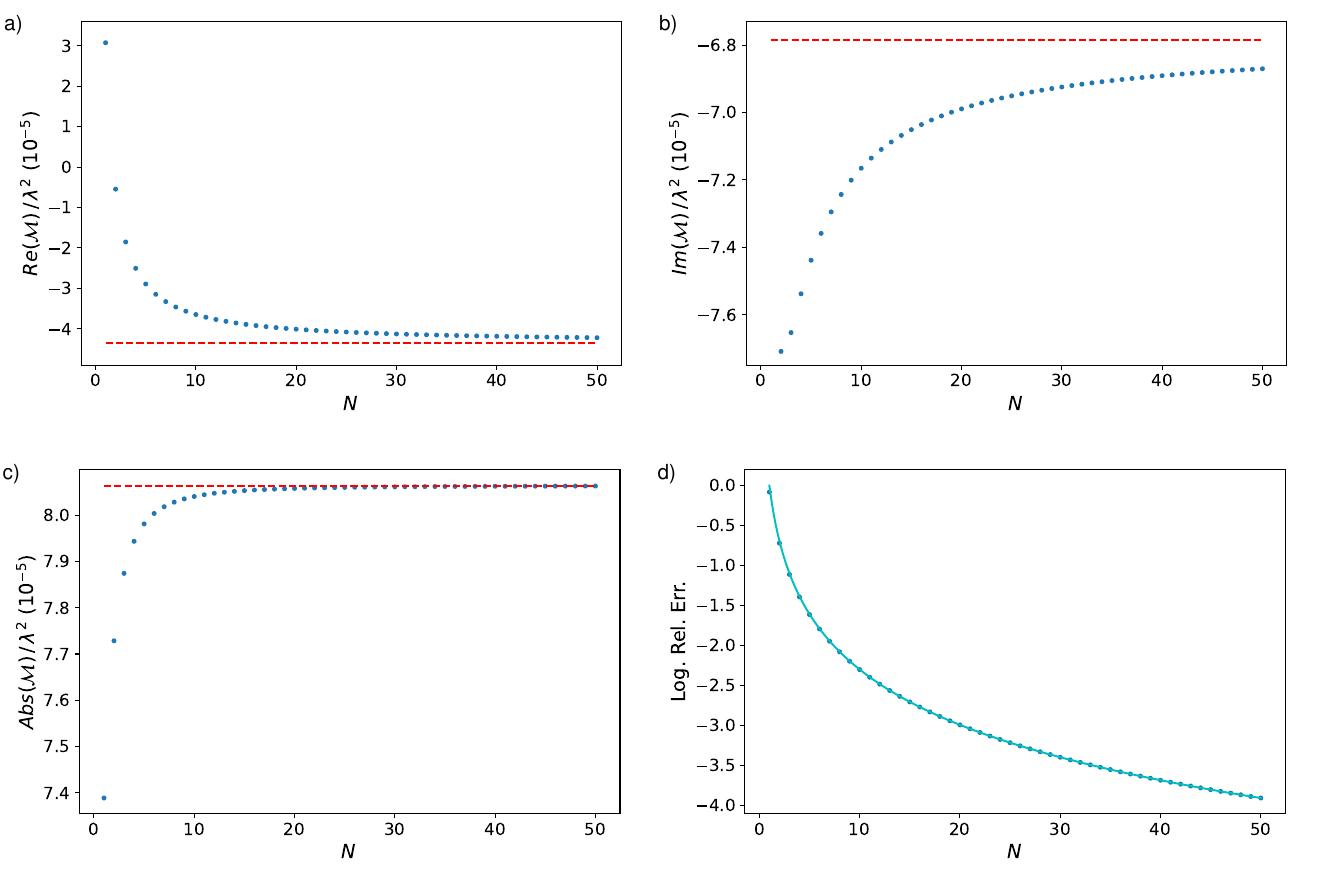}
\caption{a), b), and c) Respectively, real part, imaginary part, and absolute value of the $\mathcal{M}$ term for a pair of identical two-level UDW detectors coupled to a quantum scalar field in spacelike separated regions through a train of delta couplings approximating a Heaviside switching, as a function of the number of couplings $N$. Here we set the relevant parameters to $\gamma=1$, $d=1.2$, and $r=0.1$. The constant red dashed lines mark the values of the real part, imaginary part, and absolute value of the $\mathcal{M}$ term for the exact Heaviside switching. d) Logarithmic relative error of the $\mathcal{M}$ term for the Heaviside switching approximated by delta couplings, as a function of the number of couplings $N$. The function represented by the solid cyan line is proportional to $1/N$, marking the rate of convergence of the approximation.}
\label{Fig: 2d M}
\end{center}
\end{figure*}

\vspace{0.5cm}

\section{Conclusion}\label{Section: Conclusion}

We have shown how to efficiently approximate non-perturbatively the time evolution of particle detectors interacting with a quantum field for very general compactly supported switching functions and arbitrary smearing functions.

Specifically, we have described how to approximate the detectors' dynamics when they couple through compactly supported bounded switching functions, continuous except for maybe a finite number of points, with that of a detector coupled through a sequence of delta couplings, which can be evaluated non-perturbatively. The approximation converges at least as fast as $1/N$ (where $N$ is the number of delta pulses), and often much faster than this for regular enough switchings. 

The fast convergence of the approximation is guaranteed at all orders in perturbation theory as long as the detectors' switching and smearing functions and the field's state are `regular enough'. This, in particular, includes the cases in which the state of the field is Hadamard, and the switching and smearing functions are smooth except for maybe a finite number of points, which covers most of the relevant scenarios in both flat and curved spacetimes.

Since this approximation scheme can be easily (and efficiently) evaluated in current computers, we expect the method presented here to pave the way for a number of non-perturbative analyses of phenomena involving particle detectors and measurements in quantum field theory.

\begin{acknowledgments}	

The authors would like to thank T. Rick Perche, Bruno de S. L. Torres, and Sergi Nadal for their comments and insights. The authors would also like to acknowledge Jos\'e de Ram\'on and Jose M. S\'anchez Vel\'azquez, who were part of some early discussions related to this project. JPG acknowledges the support of a Mike and Ophelia Lazaridis Fellowship, as well as the support of a fellowship from ``La Caixa'' Foundation (ID 100010434, code LCF/BQ/AA20/11820043). EMM acknowledges support through the Discovery Grant Program of the Natural Sciences and Engineering Research Council of Canada (NSERC). EMM also acknowledges support of his Ontario Early Researcher award. Research at Perimeter Institute is supported in part by the Government of Canada through the Department of Innovation, Science and Industry Canada and by the Province of Ontario through the Ministry of Colleges and Universities.		
\end{acknowledgments}

\onecolumngrid
\appendix

\section{Proof of the approximation result}\label{Appendix: Proof of the approximation result}

In this Appendix, we prove the approximation results stated in Eqs.~\eqref{Eq: Approximation result order by order}--\eqref{Eq: Approximation result} for a single detector, and in Eqs.~\eqref{Eq: Approximation result order by order multiple bis}--\eqref{Eq: Approximation result multiple} for multiple detectors. The claim establishes that each term of the Dyson series of the time-evolved density matrix resulting from couplings with trains of sudden interactions converges (as we increase the number of sudden interactions) to the contribution of the same order for the density matrix resulting from couplings with (sufficiently regular) switching functions that the delta interactions aim to approximate, both for setups involving one or more particle detectors.  



Let us first consider a one detector scenario with the interaction Hamiltonian weight given by Eq.~\eqref{Eq: interaction Hamiltonian}, i.e., 
\begin{equation}
\hat h_\text{int}(\tau, \bm z) = \lambda \chi(\tau) \big[ F(\bm z) \,\muh_\alpha^\dagger(\tau)\,\geno^\alpha(\tau,\bm z) + \text{H.c.}\big],
\end{equation}
where the switching function $\chi$ can be either the regular switching, $\xi$, or its approximation with $N$ sudden interactions, $\chi_\xi(N)$, given in Eq.~\eqref{Eq: chi_xi}. For the regular scenario, the joint evolution of the field and the detector is given by the unitary 
\begin{equation}\label{Eq: evolution operator continuous App}
\hat U_\xi = \mathcal{T}_\tau \exp\bigg( -\ii\int\diff\maf z\, \hat h_\text{int}^\xi(\maf z)   \bigg) = \mathcal{T}_\tau \exp\bigg\{ -\ii\lambda\int\diff\tau\,\xi(\tau) \bigg[ \muh_\alpha^\dagger(\tau) \int\diff\bm z\, \sqrt{-g}\,F(\bm z) \,\geno^\alpha(\tau,\bm z) + \text{H.c.} \bigg] \bigg\},
\end{equation}
with the final state of the detector being
\begin{equation}
\rhoh_\textsc{d}^\xi = \Tr_\phi\Big(\hat U_\xi \rhoh_0 \hat U_\xi^\dagger \Big).
\end{equation}
Expanding in series of $\lambda$, we get
\begin{equation}
\hat U_\xi = \openone + \sum_{k=1}^{\infty} \hat U_\xi^{(k)} \quad \text{and} \quad \rhoh_\textsc{d}^\xi =  \rhoh_\textsc{d,0} + \sum_{k=1}^\infty \rhoh_\textsc{d}^{(k)\xi},
\end{equation}
where $\hat U_\xi^{(k)}$ and $\rhoh_\textsc{d}^{(k)\xi}$ are proportional to $\lambda^k$, and
\begin{equation}\label{Eq: k-th order continuous rho}
\rhoh_\textsc{d}^{(k)\xi} = \sum_{r+s=k} \Tr_\phi\Big( \hat U_\xi^{(r)}\rhoh_0\hat U_\xi^{(s)\dagger}  \Big).
\end{equation}
Now,
\begin{align}\label{Eq: perturbative piece cont}
\Tr_\phi\Big( \hat U_\xi^{(r)}\rhoh_0\hat U_\xi^{(s)\dagger} \Big) & = (-\ii)^r\ii^s \lambda^k \int\diff\tau_1\hdots\diff\tau_r\diff\tau'_1\hdots\diff\tau'_s\,\prod_{p=1}^{r-1}\theta(\tau_{p+1}-\tau_p) \prod_{q=1}^{s-1}\theta(\tau'_{q+1}-\tau'_q)  \\
& \phantom{=\;} \times \prod_{l=1}^{r}\xi(\tau_l)\,\prod_{m=1}^{s}\xi(\tau'_m) \,\Big[W_{\bm\alpha'}^{\bm\alpha}(\tau'_1,\hdots,\tau'_s;\tau_1,\hdots,\tau_r; F) \, \muh_{\alpha_r}^\dagger(\tau_r)\hdots\muh_{\alpha_1}^\dagger(\tau_1)\rhoh_\textsc{d,0}\muh^{\alpha'_1}(\tau'_1)\hdots\muh^{\alpha'_s}(\tau'_s) \nonumber \\[-0.2cm]
&\phantom{==============} + W^{\bm\alpha'}_{\bm\alpha}\!(\tau'_1,\hdots,\tau'_s;\tau_1,\hdots,\tau_r) \,\muh^{\alpha_r}(\tau_r)\hdots\muh^{\alpha_1}(\tau_1)\rhoh_\textsc{d,0}\muh_{\alpha'_1}^\dagger(\tau'_1)\hdots\muh_{\alpha'_s}^\dagger(\tau'_s) \nonumber \\[0.05cm]
& \phantom{==============} + W_{\bm\alpha,\bm\alpha'}(\tau'_1,\hdots,\tau'_s;\tau_1,\hdots,\tau_r)\, \muh^{\alpha_r}(\tau_r)\hdots\muh^{\alpha_1}(\tau_1)\rhoh_\textsc{d,0}\muh^{\alpha'_1}(\tau'_1)\hdots\muh^{\alpha'_s}(\tau'_s) \nonumber \\[0.05cm]
& \phantom{==============} + W^{\bm\alpha,\bm\alpha'}\!(\tau'_1,\hdots,\tau'_s;\tau_1,\hdots,\tau_r) \,\muh_{\alpha_r}^\dagger(\tau_r)\hdots\muh_{\alpha_1}^\dagger(\tau_1)\rhoh_\textsc{d,0}\muh_{\alpha'_1}^\dagger(\tau'_1)\hdots\muh_{\alpha'_s}^\dagger(\tau'_s) \Big], \nonumber
\end{align}
where we denoted
\begin{align}\label{Eq: W_k def}
W_{\bm\alpha'}^{\bm\alpha}(t'_1,\hdots,t'_s;t_1,\hdots,t_r; F) & \coloneqq \Tr_\phi\big[\rhoh_\phi\,\geno_{\alpha'_1}^\dagger(t'_1,F)\hdots\geno_{\alpha'_s}^\dagger(t'_s,F)\geno^{\alpha_r}(t_r,F)\hdots\geno^{\alpha_1}(t_1,F)\big], \\
W_{\bm\alpha}^{\bm\alpha'}\!(t'_1,\hdots,t'_s;t_1,\hdots,t_r; F) & \coloneqq \Tr_\phi\big[\rhoh_\phi\,\geno^{\alpha'_1}(t'_1,F)\hdots\geno^{\alpha'_s}(t'_s,F)\geno_{\alpha_r}^\dagger(t_r,F)\hdots\geno_{\alpha_1}^\dagger(t_1,F)\big], \\
W_{\bm\alpha,\bm\alpha'}(t'_1,\hdots,t'_s;t_1,\hdots,t_r; F) & \coloneqq \Tr_\phi\big[\rhoh_\phi\,\geno_{\alpha'_1}^\dagger(t'_1,F)\hdots\geno_{\alpha'_s}^\dagger(t'_s,F)\geno_{\alpha_r}^\dagger(t_r,F)\hdots\geno_{\alpha_1}^\dagger(t_1,F)\big], \\
W^{\bm\alpha,\bm\alpha'}\!(t'_1,\hdots,t'_s;t_1,\hdots,t_r; F) & \coloneqq \Tr_\phi\big[\rhoh_\phi\,\geno^{\alpha'_1}(t'_1,F)\hdots\geno^{\alpha'_s}(t'_s,F)\geno^{\alpha_r}(t_r,F)\hdots\geno^{\alpha_1}(t_1,F)\big],
\end{align}
and
\begin{equation}
\geno^\alpha(t,F) \coloneqq \int\diff\bm z\,\sqrt{-g}\, F(\bm z) \,\geno^\alpha(t,\bm z).
\end{equation}
We can rewrite Eq.~\eqref{Eq: perturbative piece cont} as
\begin{equation}\label{Eq: perturbative piece cont rewritten}
\Tr_\phi\Big(  \hat U_\xi^{(r)} \rhoh_0\hat U_\xi^{(s)\dagger} \Big) = (-\ii)^r \ii^s \lambda^k \int_{\Delta_r\times\Delta'_s}\diff\bm\tau\,\diff\bm\tau'\, \Xi(\bm\tau;\bm\tau')\, \hat \omega(\bm\tau;\bm\tau';F).
\end{equation}
Here,
\begin{align}
\Delta_r & \coloneqq \{\bm\tau =(\tau_1,\hdots,\tau_r)\,:\,\tau_1\in[0,T],\,\tau_i\in[0,\tau_{i-1}],\,\forall\,i=2,\hdots,r\}, \\
\Delta'_s & \coloneqq \{\bm\tau' \!=(\tau'_1,\hdots,\tau'_s)\,:\,\tau'_1\in[0,T],\,\tau_i\in[0,\tau_{i-1}],\,\forall\,i=2,\hdots,s\},
\end{align}
are triangular domains; the function $\Xi$ is an abbreviation,
\begin{equation}\label{Eq: Def Xi}
\Xi(\bm\tau;\bm\tau') \coloneqq \prod_{l=1}^{r}\xi(\tau_l)\,\prod_{m=1}^{s}\xi(\tau'_m),
\end{equation}
and $\hat\omega$ is the operator corresponding to the expression in square brackets of Eq.~\eqref{Eq: perturbative piece cont}. 

On the other hand, if we consider the evolution produced by the coupling through the train of sudden interactions described by $\chi_\xi(N)$, from Eqs.~\eqref{Eq: evolution operator} and~\eqref{Eq: chi_xi} we have that the joint field-detector time-evolution is given by
\begin{equation}
\hat U(N) = \mathcal{T}_\tau \exp\bigg\{ -\ii\lambda \frac{T}{N}\sum_{j=1}^N \xi\bigg( 
\frac{j-1/2}{N}T \bigg)  \bigg[ \muh_\alpha^\dagger\bigg( \frac{j-1/2}{N}T \bigg)\int\diff\bm z\,\sqrt{-g_j}\, F(\bm z) \geno^\alpha\bigg(\frac{j-1/2}{N},\bm z \bigg) + \text{H.c.}\bigg] \bigg\},
\end{equation}
with the final state of the detector thus being
\begin{equation}
\rhoh_\textsc{d}(N) = \Tr_\phi\Big( \hat U(N)\, \rhoh_0 \,\hat U(N)^\dagger \Big).
\end{equation}
Under the same conditions as before, we write
\begin{equation}
\hat U = \openone + \sum_{k=1}^\infty \hat U^{(k)}(N) \quad \text{and} \quad \rhoh_\textsc{d}(N) = \rhoh_\textsc{d,0} + \sum_{k=1}^\infty \rhoh_\textsc{d}^{(k)}(N),
\end{equation}
where again $\hat U^{(k)}(N)$ and $\rhoh_\textsc{d}^{(k)}(N)$ are proportional to $\lambda^k$, and
\begin{equation}\label{Eq: k-th order deltas rho}
\rhoh_\textsc{d}^{(k)}(N) = \sum_{r+s=k} \Tr_\phi\Big( \hat U^{(r)}(N)\,\rhoh_0\, \hat U^{(s)}(N)^{\dagger} \Big).
\end{equation}
In this case, for $r+s=k$, we have
\begin{align}\label{Eq: perturbative piece deltas}
\Tr_\phi\Big( \hat U^{(r)}(N)\, \rhoh_0 \,\hat U^{(s)}(N)^{\dagger} \Big) & = (-\ii)^r \ii^s \frac{\lambda^k T^k}{N^k} \sum_{\bm j \in \mathcal{J}_r}\sum_{\bm j' \in \mathcal{J}'_s} \prod_{n=1}^r \xi\bigg( \frac{j_n-1/2}{N} T\bigg) \prod_{m=1}^s \xi\bigg( \frac{j'_m - 1/2}{N}T \bigg) \\
& \times \bigg[ W_{\bm\alpha'}^{\bm\alpha}\bigg( \frac{j'_1 - 1/2}{N}T,\hdots,\frac{j'_s-1/2}{N}T;\frac{j_1-1/2}{N}T,\hdots,\frac{j_r-1/2}{N}T; F  \bigg) \nonumber \\
&\phantom{==}\times\muh_{\alpha_r}^\dagger\bigg( \frac{j_r-1/2}{N}T \bigg)\hdots\muh_{\alpha_1}^\dagger\bigg( \frac{j_1-1/2}{N}T \bigg)\,\rhoh_{\textsc{d,0}} \,\muh^{\alpha'_1}\bigg( \frac{j'_1-1/2}{N}T \bigg)\hdots \muh^{\alpha'_s}\bigg( \frac{j'_s-1/2}{N}T \bigg) \nonumber \\
&\phantom{==} + W_{\bm\alpha}^{\bm\alpha'}\bigg( \frac{j'_1 - 1/2}{N}T,\hdots,\frac{j'_s-1/2}{N}T;\frac{j_1-1/2}{N}T,\hdots,\frac{j_r-1/2}{N}T; F  \bigg) \nonumber \\
&\phantom{==}\times\muh^{\alpha_r}\bigg( \frac{j_r-1/2}{N}T \bigg)\hdots\muh^{\alpha_1}\bigg( \frac{j_1-1/2}{N}T \bigg)\,\rhoh_{\textsc{d,0}} \,\muh_{\alpha'_1}^\dagger\bigg( \frac{j'_1-1/2}{N}T \bigg)\hdots \muh_{\alpha'_s}^\dagger\bigg( \frac{j'_s-1/2}{N}T \bigg) \nonumber \\
&\phantom{==} + W_{\bm\alpha,\bm\alpha'}\bigg( \frac{j'_1 - 1/2}{N}T,\hdots,\frac{j'_s-1/2}{N}T;\frac{j_1-1/2}{N}T,\hdots,\frac{j_r-1/2}{N}T; F  \bigg) \nonumber \\
&\phantom{==}\times\muh^{\alpha_r}\bigg( \frac{j_r-1/2}{N}T \bigg)\hdots\muh^{\alpha_1}\bigg( \frac{j_1-1/2}{N}T \bigg)\,\rhoh_{\textsc{d,0}} \,\muh^{\alpha'_1}\bigg( \frac{j'_1-1/2}{N}T \bigg)\hdots \muh^{\alpha'_s}\bigg( \frac{j'_s-1/2}{N}T \bigg) \nonumber \\
&\phantom{==} + W^{\bm\alpha,\bm\alpha'}\!\bigg( \frac{j'_1 - 1/2}{N}T,\hdots,\frac{j'_s-1/2}{N}T;\frac{j_1-1/2}{N}T,\hdots,\frac{j_r-1/2}{N}T; F  \bigg) \nonumber \\
&\phantom{==}\times\muh_{\alpha_r}^\dagger\bigg( \frac{j_r-1/2}{N}T \bigg)\hdots\muh_{\alpha_1}^\dagger\bigg( \frac{j_1-1/2}{N}T \bigg)\,\rhoh_{\textsc{d,0}} \,\muh_{\alpha'_1}^\dagger\bigg( \frac{j'_1-1/2}{N}T \bigg)\hdots \muh_{\alpha'_s}^\dagger\bigg( \frac{j'_s-1/2}{N}T \bigg) \bigg], \nonumber
\end{align}
where we have used the same notation as in Eq.~\eqref{Eq: perturbative piece cont}, and
\begin{align}
\mathcal{J}_r & \coloneqq \{ \bm j = (j_1,\hdots,j_r) \,:\, j_1 \in \{1,\hdots,N\},\, j_i \in \{1,\hdots,j_{i-1}\},\,\forall\,i=2,\hdots,r  \},  \\
\mathcal{J}'_s & \coloneqq \{ \bm j' = (j'_1,\hdots,j'_s) \,:\, j'_1 \in \{1,\hdots,N\},\, j'_i \in \{1,\hdots,j'_{i-1}\},\,\forall\,i=2,\hdots,s  \}.
\end{align}
Using the notation employed in Eq.~\eqref{Eq: perturbative piece cont rewritten}, we can rewrite Eq.~\eqref{Eq: perturbative piece deltas} as
\begin{equation}\label{Eq: perturbative piece deltas rewritten}
\Tr_\phi\Big(  \hat U^{(r)}(N)\, \rhoh_0\,\hat U^{(s)}(N)^{\dagger} \Big) = (-\ii)^r \ii^s \frac{\lambda^k T^k}{N^k} \sum_{(\bm j,\bm j') \in \mathcal{J}_r \times \mathcal{J}'_s} \Xi(\bm\tau_{\bm{j}};\bm\tau'_{\bm{j'}}) \, \hat\omega(\bm\tau_{\bm j};\bm\tau'_{\bm{j}'};F),
\end{equation}
where
\begin{equation}
\bm\tau_{\bm j} \coloneqq \bigg( \frac{j_1-1/2}{N}T,\hdots,\frac{j_r-1/2}{N}T \bigg) \quad \text{and} \quad \bm\tau'_{\bm j'} \coloneqq \bigg( \frac{j'_1-1/2}{N}T,\hdots,\frac{j'_s-1/2}{N}T \bigg). 
\end{equation}

To prove the approximation result presented in Sec.~\ref{Section: Single detector}, we follow a generalization of the strategy used by C. K. Chui in~\cite{Chui1969,Chui1971}. First, let us define $H_{ab}^{(r,s)}:\R{r}\times\R{s}\to\C{}$ as
\begin{equation}
H_{ab}^{(r,s)}(\bm \tau;\bm \tau') = (-\ii)^r \ii^s \Xi(\bm\tau;\bm\tau') \,\omega_{ab}(\bm\tau;\bm\tau';F) \, \textsc{I}_{\Delta_r \times \Delta'_s},
\end{equation}
where the dependence on $r$ and $s$ for $\Xi$ and $\hat\omega$ manifests itself in the vectors $\bm \tau$ and $\bm \tau'$ having $r$ and $s$ components, respectively. Here, $\omega_{ab}$ is the $(a,b)$-th matrix component of $\hat \omega$ in some basis of the detector Hilbert space $\hilb_\textsc{d}$, and $\textsc{I}_\textsc{x}$ is the indicator function for the set $\textsc{X}$, i.e., a function that equals 1 when the argument belongs to $\textsc{X}$, and 0 otherwise. With this definition, from Eq.~\eqref{Eq: perturbative piece cont rewritten} we have that
\begin{equation}
\bigg[\Tr_\phi\Big(  \hat U_\xi^{(r)} \rhoh_0\hat U_\xi^{(s)\dagger} \Big)\bigg]_{ab} = \lambda^k \int_{[0,T]^k} \diff\bm z\, H_{ab}^{(r,s)}(\bm z),
\end{equation}
where the subindices $ab$ again denote the $(a,b)$-th matrix element of the trace in the chosen basis of $\hilb_\textsc{d}$. It is worth remarking that the presence of the indicator function in the definition of $H_{ab}$ allowed us to have here the hypercube $[0,T]^k$ as integration domain, instead of the product of triangular domains of Eq.~\eqref{Eq: perturbative piece cont rewritten}. Meanwhile, from Eq.~\eqref{Eq: perturbative piece deltas rewritten} we also have that
\begin{equation}
\bigg[ \Tr_\phi\Big(  \hat U^{(r)}(N)\, \rhoh_0\,\hat U^{(s)}(N)^{\dagger} \Big) \bigg]_{ab} = \lambda^k R_{N}^{(k)}\big[H_{ab}^{(r,s)}\big],
\end{equation}
where $R_N^{(k)}$ is the Riemann sum associated with the uniform partition of $[0,T]^{k}$ into $N^{k}$ identical cubes with tags placed in their centres, i.e., for an arbitrary $h:\R{k}\to\C{}$,
\begin{equation}
R_{N}^{(k)}[h] \coloneqq \sum_{\bm j} h\bigg( \frac{j_1-1/2}{N}T,\hdots,\frac{j_k-1/2}{N}T  \bigg) \cdot \frac{T^k}{N^k},
\end{equation}
with each component of $\bm j$ running over $\{1,\hdots,N\}$. Let us now define $v_N^{(k)} : \R{k} \to \C{}$ as
\begin{equation}
v_{N}^{(k)}(\bm z) = \sum_{\bm j}T^k \prod_{i=1}^{k}\textsc{I}_{\big[\frac{j_i-1/2}{N},1\big]}\bigg( \frac{z_i}{T} \bigg) - N^k \prod_{i=1}^k z_i,
\end{equation}
where again each component of $\bm j$ runs over $\{1,\hdots,N\}$. With this definition, we can write
\begin{equation}
R_N^{(k)}\big[H_{ab}^{(r,s)}\big] - \int_{[0,T]^k}\diff\bm z\, H_{ab}^{(r,s)}(\bm z) = \frac{1}{N^k} \int_{[0,T]^k} \diff v_N^{(k)}\,H_{ab}^{(r,s)},
\end{equation}
where the right-hand side should be understood as a $k$-dimensional Riemann-Stieltjes integral of $H_{ab}^{(r,s)}$ over $[0,T]^k$ with respect to $v_N^{(k)}$, as defined in Appendix~\ref{Appendix: The n-dimensional Riemann-Stieltjes integral}. Now, $H_{ab}^{(r,s)}$ is a product of functions including $\Xi$ (defined in Eq.~\eqref{Eq: Def Xi}), which is zero whenever one of its arguments is $0$ or $T$. In particular, $H_{ab}^{(r,s)}$ is zero along the boundary of $[0,T]^k$, so by Lemma~\ref{Lemma: somewhat integration by parts} in Appendix~\ref{Appendix: The n-dimensional Riemann-Stieltjes integral}, 
\begin{equation}
R_N^{(k)}\big[H_{ab}^{(r,s)}\big] - \int_{[0,T]^k}\diff\bm z\, H_{ab}^{(r,s)}(\bm z) = \frac{(-1)^k}{N^k} \int_{[0,T]^k} \diff H_{ab}^{(r,s)}\,v_N^{(k)}.
\end{equation}
Since $v_N^{(k)}$ is bounded in $[0,T]^k$, with bound
\begin{equation}
\big\|v_N^{(k)}\big\|_\infty \leq \frac{k T^k}{2} N^{k-1},
\end{equation}
assuming $H_{ab}^{(r,s)}$ is of bounded variation, we can use Lemma~\ref{Lemma: bound for the integral} in Appendix~\ref{Appendix: The n-dimensional Riemann-Stieltjes integral} to conclude that
\begin{equation}\label{Eq: bound of error}
\bigg| R_N^{(k)}\big[H_{ab}^{(r,s)}\big] - \int_{[0,T]^k}\diff\bm z\, H_{ab}^{(r,s)}(\bm z) \bigg| \leq \frac{k \,\tilde{C}_{ab}^{(r,s)}}{2N} T^k \equiv \frac{C_{ab}^{(r,s)}}{N},
\end{equation}
where $\tilde{C}_{ab}^{(r,s)}$ is the variation of $H_{ab}^{(r,s)}$ in $[0,T]^k$.  If $H_{ab}^{(r,s)}$ is of bounded variation for all pairs $(r,s)$ such that $r+s=k$, from Eqs.~\eqref{Eq: k-th order continuous rho} and~\eqref{Eq: k-th order deltas rho} we get that
\begin{equation}
\big|[\rhoh_\textsc{d}^{(k)}(N)]_{ab} - [\rhoh_\textsc{d}^{(k)\xi}]_{ab}\big| \leq \lambda^k \sum_{k=r+s} \frac{C_{ab}^{(r,s)}}{N} = \lambda^k \frac{C_{ab}^{(k)}}{N},
\end{equation}
for some $C_{ab}^{(k)}$, which proves Eq.~\eqref{Eq: Approximation result order by order bis}. This also implies
\begin{equation}\label{Eq: approx result component by component}
\lim_{N\to\infty}[\rhoh_\textsc{d}^{(k)}(N)]_{ab} = [\rhoh_\textsc{d}^{(k)\xi}]_{ab},
\end{equation}
which proves Eq.~\eqref{Eq: Approximation result order by order}. In particular, if Eq.~\eqref{Eq: bound of error} holds at all orders, Eqs.~\eqref{Eq: Approximation result bis} and~\eqref{Eq: Approximation result} are also satisfied for all values of the coupling strength for which the Dyson series converges.

In order to arrive at the bound given in Eq.~\eqref{Eq: bound of error}, we needed to assume that $H_{ab}^{(r,s)}$ is of bounded variation as a function of its (time) arguments. Since at the end of the day we want to use this non-perturbative tool for physical scenarios, it is worth discussing how exotic this condition is. 

Notice first that $H_{ab}^{(r,s)}$ depends on $\Xi$, which is a product of the values of the switching functions at different times, and the matrix elements $\omega_{ab}$, which involve a linear combination of smeared $k$-point correlation functions of the field operator $\geno^\alpha$ and the corresponding matrix element of an operator that consists of products of $\muh$ with $\rhoh_{\textsc{d},0}$. If the dependence of $\muh$ on time in the interaction picture only comes from the free evolution of the detector, as it is often the case, then we can assume that it is a smooth function of time. That way, whether $H_{ab}^{(r,s)}$ is sufficiently regular for Eq.~\eqref{Eq: bound of error} to hold or not ends up depending on the switching function (through $\Xi$) and the state of the field (through the smeared $k$-point correlation functions for the operator $\geno$). By lemma~\ref{Lemma: bounded variation} in Appendix~\ref{Appendix: The n-dimensional Riemann-Stieltjes integral}, one sufficient (but not necessary) condition for $H_{ab}^{(r,s)}$ to be of bounded variation is that it is of class $\mathcal{C}^{k}$ in $[0,T]^k$. In particular, this can be straightforwardly extended to the case in which $H_{ab}^{(r,s)}$ is bounded, and $\mathcal{C}^k$ except for maybe a finite number of points in $[0,T]^k$. 

Now, the (continuous) switching functions typically employed in particle detector setups (e.g., Heaviside, Gaussian, and Lorentzian couplings) are smooth except for maybe a finite number of points. Therefore, when $\xi$ is one of these common switching functions, it would suffice that the field's $k$-point correlation functions be of class $\mathcal{C}^k$ in $[0,T]^k$ for $H_{ab}^{(r,s)}$ to be bounded and of class $\mathcal{C}^k$ except for maybe a finite number of points, for all $(r,s)$ such that $r+s=k$. 

Let us consider, for instance, the vacuum of a real scalar field in $(1+3)$-dimensional Minkowski spacetime. For this case, the odd correlation functions vanish, while the even correlation functions can all be expressed as sums of products of the two-point function
\begin{equation}
w_2(\maf x_1,\maf x_2) = \int \frac{\diff \bm k}{(2\pi)^3\, 2 |\bm k|} \, e^{\ii \maf k \cdot (\maf x_1 - \maf x_2)}.
\end{equation}
If we have a linear coupling with the field amplitude, as in the Unruh-DeWitt model~\cite{Unruh,DeWitt}, \mbox{$\geno=\phih$}, and the relevant $k$-point functions can be given in terms of the smeared two-point function  
\begin{align}
W(t;t';F) = \int \diff\bm x \,\diff\bm x'\,F(\bm x) F(\bm x') \,w_2(\maf x_1,\maf x_2) = \int \frac{\diff\bm k}{(2\pi)^n \, 2 |\bm k|}\, |\tilde{F}(\bm k)|^2 \, e^{-\ii|k|(t-t')},
\end{align}
for a given smearing function $F$. One very popular choice of smearing function in the literature is the Gaussian of variance $\sigma^2$,
\begin{equation}
G(\bm x) =\frac{1}{\sqrt{\pi^3}\sigma^3}e^{-\bm x^2/\sigma^2}.
\end{equation}
With a Gaussian smearing, a straightforward calculation using spherical coordinates yields
\begin{equation}
W(t;t';G) = \frac{1}{4\pi^2} \int_0^\infty \diff |\bm k| \, |\bm k| e^{-|\bm k|^2\sigma^2/2 - \ii |\bm k| (t-t')} = \frac{1 - \sqrt{\frac{\pi}{2 \sigma^2}}(t-t')\, e^{-(t-t')^2/2\sigma^2}\big[ \ii + \operatorname{erfi}\big(\frac{t-t'}{\sqrt{2}\sigma}  \big) \big]}{ \sigma^2},
\end{equation}
which is a smooth function of $t$ and $t'$. Another common choice is the hard sphere of radius $R>0$,
\begin{equation}
S(\bm x) = \theta(R-|\bm x|).
\end{equation}
Notice that its Fourier transform
\begin{equation}
\tilde{S}(\bm k) = \frac{4\pi}{k^3}\big[ \sin(kR) - kR\cos(kR) \big]
\end{equation}
converges to a constant when $k\to 0$, and behaves like $1/k^2$ as $k\to\infty$, so that
\begin{equation}
W(t;t';S) = 4 \int_0^\infty \diff |\bm k| \, \frac{\big[ \sin(|\bm k|R) - |\bm k|R\cos(|\bm k|R) \big]^2}{|\bm k|^5} \, e^{-\ii |\bm k| (t-t')}
\end{equation}
is convergent. It can be shown, moreover, that it is of bounded variation, even though in this case
\begin{equation}
\frac{\partial^2}{\partial t \partial t'}W(t;t';S) = 4 \int_0^\infty \diff |\bm k| \, \frac{\big[ \sin(|\bm k|R) - |\bm k|R\cos(|\bm k|R) \big]^2}{|\bm k|^3} \, e^{-\ii |\bm k| (t-t')}
\end{equation}
fails to converge and therefore we cannot use Lemma~\ref{Lemma: bounded variation} to reach that conclusion. Thus, we see that for two of the most common smearing functions\footnote{For a more detailed discussion of when and how the smearing function regularizes the two-point distribution, see~\cite{Louko2006}.}, the vacuum of a real scalar field in $(1+3)$-dimensional Minkowski satisfies the condition for Eq.~\eqref{Eq: approx result component by component} to hold, \textit{at all orders} in perturbation theory. Since, for two or more spatial dimensions\footnote{Also for one spatial dimension, assuming that the infrared divergences are properly regularized.}, the possible non-smoothness of the field correlation functions can only arise from the ultraviolet behaviour of the field state, we can argue that any Hadamard state, both in flat and curved spacetime (barring pathological geometries), will essentially have the same ultraviolet behaviour as the Minkowski vacuum, and therefore should in principle satisfy the assumptions of the approximation result as well.

Finally, we can turn to the multiple detector scenario. The proof proceeds in just the same way as for the single detector case. For the sake of brevity and simplicity we will not give its details here, but we argue that the procedure is completely analogous. Specifically, the similar structure of Eqs.~\eqref{Eq: each unitary single detector} and~\eqref{Eq: each unitary multiple detectors} is revealing of how close both scenarios are from the technical point of view. The case of multiple detectors will in general involve products of different switching functions and different detector operators, as well as field correlation functions associated to potentially different field operators, smeared over different spatial profiles. However, one can give homologous definitions of $\Xi$ (cf. Eq.~\eqref{Eq: Def Xi}) and $\hat\omega$ (cf. Eq.~\eqref{Eq: perturbative piece cont}) that lead to the multiple detectors version of Eq.~\eqref{Eq: perturbative piece cont rewritten}. The steps that follow are exactly those followed above for the one detector scenario, adapted to the new integration regions, which for the case of multiple detectors are general rectangles, instead of the hypercubes we dealt with in the case of a single detector.

\section{The $n$-dimensional Riemann-Stieltjes integral}\label{Appendix: The n-dimensional Riemann-Stieltjes integral}

Even though the Riemann-Stieltjes integral is a basic tool in real analysis, its generalization to the $n$-dimensional case is not as popular in classical texts. Because of that, in this appendix we present a very brief introduction to the $n$-dimensional Riemann-Stieltjes integral---or rather, its generalization to functions of complex argument---and we prove the three lemmas we referred to in Appendix~\ref{Appendix: Proof of the approximation result}. 

Let \mbox{$f,g :  \R{n} \to \C{}$}, and let \mbox{$R= [x_1,y_1] \times \hdots \times [x_n,y_n] \subset \R{n}$} be a closed rectangle in $\R{n}$. Following~\cite{Hildebrandt1963}, we define
\begin{equation}
\Delta^{(n)}_R g = \sum_{\bm z} (-1)^{\#x(\bm z)} g(\bm z),
\end{equation}
where the sum runs over all \mbox{$\bm z = (z_1,\hdots,z_n) \in \R{n}$} such that $z_i= x_i$ or $y_i$, for each \mbox{$i=1,\hdots,n$}, and $\#x(\bm z)$ is the number of components of $\bm z$ that are equal to $x_i$ and not $y_i$. This definition is motivated by the fact that for the function \mbox{$p(\bm z) = \prod_{i=1}^n z_i$}, we have 
\begin{equation}
\Delta^{(n)}_R p = \prod_{i=1}^{n} (y_i-x_i) = \operatorname{Vol}(R).
\end{equation}
Now, for each $i=1,\hdots,n$, consider a partition of the interval $[x_i,y_i]$, 
\begin{equation}
P_i=\{t_0^{(i)}=x_i,\hdots,t_{M_i}^{(i)}=y_i\},
\end{equation}
with a set of tags
\begin{equation}
T_i = \{u_1,\hdots,u_{M_i}\,:\,u_j^{(i)} \in [t_{j-1}^{(i)},t_j^{(i)}],\,\forall\,j=1,\hdots,M_{i}\}.
\end{equation}
All these partitions together induce an $n$-dimensional rectangular partition \mbox{$P \equiv P_1\times\hdots\times P_n$} of the rectangle $R$, formed by subrectangles 
\begin{equation}
R_{\bm j} = [t_{j_1-1}^{(1)},t_{j_1}^{(1)}] \times \hdots \times [t_{j_n-1}^{(n)},t_{j_n}^{(n)}] \subset R
\end{equation}
that can only overlap on their edges. Each subrectangle also has an associated tag \mbox{$\bm u_{\bm j}=(u_{j_1},\hdots,u_{j_n})$}, and we can define the norm of $P$ as the maximum among the norms of each $P_i$, i.e., 
\begin{equation}
\norm{P}\coloneqq  \max_{i=1,\hdots,n}\norm{P_i}=\max\{|t^{(i)}_{j}-t^{(i)}_{j-1}|\,:\,j=1,\hdots,M_i,\,i=1,\hdots,n\}.
\end{equation}
We define the Riemann-Stieltjes sum associated with $f$, $g$, and the $n$-dimensional tagged partition $P$ of the rectangle $R$ as
\begin{equation}\label{Eq: RS sum}
S_R(f,g;P) = \sum_{\bm j} f(\bm u_{\bm j})\, \Delta^{(n)}_{R_{\bm j}} g.
\end{equation}
Hence, we say that $f$ is Riemann-Stieltjes integrable with respect to $g$ in $R$ when there exists $I\in\C{}$ such that, for every $\epsilon > 0$, there is some $\delta > 0$ for which 
\begin{equation}
|S_R(f,g;P)-I| < \epsilon
\end{equation}
holds whenever $\norm{P} < \delta$. In that case, we call $I$ the Riemann-Stieltjes integral of $f$ with respect to $g$ over $R$,
\begin{equation}
I \equiv \int_R \diff g \, f.
\end{equation}

\begin{lemma}\label{Lemma: somewhat integration by parts}
Let $f,g:\R{n} \to \C{}$, and let $R= [x_1,y_1] \times \hdots \times [x_n,y_n]$ be a rectangle in $\R{n}$, such that $f$ is integrable over $R$ with respect to $g$, and either $f=0$ or $g=0$ on the boundary $\partial R$. Then, $g$ is also integrable over $R$ with respect to $f$, and
\begin{equation}
\int_R \diff f \, g = (-1)^{n} \int_R \diff g\, f.
\end{equation}
\end{lemma}
\begin{proof}
Given $\epsilon>0$, since $f$ is integrable over $R$ with respect to $g$, there exists $\delta>0$ such that if \mbox{$P=P_1\times \hdots\times P_n$} is an $n$-dimensional tagged partition of $R$ with $\norm{P}<\delta$, then
\begin{equation}\label{Eq: integrability of f}
\bigg| S_R(f,g;P) - \int_R \diff g\, f \bigg| < \epsilon.
\end{equation}
Consider now a rectangular partition of $R$, \mbox{$P=P_1\times\hdots\times P_n$}, such that $\norm{P}<\delta/2$. For each one-dimensional partition \mbox{$P_i=\{t_0^{(i)}=x_i,\hdots,t_{M_i}^{(i)}=y_i\}$} of the interval $[x_i,y_i]$, with tags \mbox{$\{u_1,\hdots,u_{M_i}\}$}, consider another partition
\begin{equation}
\bar{P}_i \coloneqq \{t_0^{(i)}=x_i,u_1,\hdots,u_{M_{i}},t_{M_i}=y_i\},
\end{equation}
with set of tags $\bar{T}_i=P_i$, i.e., 
\begin{equation}
\bar{T}_i = \{v_1,\hdots,v_{M_i},v_{M_{i}+1}\,:\,v_1=t_0,\,v_j=t_{j-1},\forall\,j>1\}.
\end{equation}
Notice in particular that the set of tags of $\bar{P}_i$ includes the limits of the interval, $x_i$ and $y_i$. We can then define a new partition $\bar{P}=\bar{P}_1\times\hdots\times\bar{P}_n$, with set of tags $\bar{T}=\bar{T}_1\times\hdots\times\bar{T}_n$, which by construction satisfies $\norm{P}<\delta$. Moreover, a simple calculation shows that
\begin{equation}\label{Eq: equality with boundary terms}
S(f,g;\bar{P}) = (-1)^{n} S(g,f; P) + \text{`Boundary terms'}, 
\end{equation}
where the boundary terms only include sums of products of $f$ and $g$ at points of the boundary $\partial R$. Since either $f=0$ or $g=0$ on the boundary, the associated terms cancel out. Eqs.~\eqref{Eq: integrability of f} and~\eqref{Eq: equality with boundary terms} thus allow us to conclude that
\begin{equation}
\bigg| (-1)^n S_R(g,f;P) - \int_R \diff g\, f \bigg| < \epsilon,
\end{equation}
and therefore that $g$ is integrable over $R$ with respect to $f$, with
\begin{equation}
\int_R \diff f\, g = (-1)^n \int_R \diff g\, f,
\end{equation}
as claimed.
\end{proof}

\begin{lemma}\label{Lemma: bound for the integral}
Let $f,g:\R{n}\to\C{}$, and let $R= [x_1,y_1] \times \hdots \times [x_n,y_n]$ be a rectangle in $\R{n}$. If $f$ is bounded in $R$ and $g$ is of bounded variation in $R$, then we have that
\begin{equation}
\bigg|\int_R \diff g\, f\bigg| \leq \norm{f}_\infty \operatorname{Var}_{R}\,(g),
\end{equation}
where 
\begin{equation}
\norm{f}_\infty = \sup_{R}|f|
\end{equation}
is the optimal bound for $f$ in $R$, and
\begin{equation}
\operatorname{Var}_R\,(g) \coloneqq \sup_{P} \sum_{\bm j} |\Delta_{R_{\bm j}}^{(n)}g|
\end{equation}
is the variation of $g$ in $R$, with the supremum taken over all rectangular partitions $P$ of the rectangle $R$, and the notation used as in Eq.~\eqref{Eq: RS sum}.
\end{lemma}
\begin{proof}
Since $f$ is integrable over $R$ with respect to $g$, for every $\epsilon>0$ there exists $\delta>0$ such that
\begin{equation}\label{Eq: epsilon inequality}
\bigg|S_R(f,g;P) - \int_R\diff g\, f\bigg| < \epsilon
\end{equation}
whenever $\norm{P}<\delta$. Now,
\begin{equation}
|S_R(f,g;P)| = \Big|\sum_{\bm j} f(\bm u_{\bm j}) \,\Delta_{R_{\bm j}}^{(n)}g \Big| \leq \sum_{\bm j} |f(\bm u_{\bm j})| \,|\Delta_{R_{\bm j}}^{(n)}g| \leq \norm{f}_{\infty} \sum_{\bm j}|\Delta_{R_{\bm j}}^{(n)}g| \leq \norm{f}_{\infty} \operatorname{Var}_{R}(g), 
\end{equation}
and
\begin{equation}
\bigg|S_R(f,g;P) - \int_R\diff g\, f\bigg| \geq \bigg| \int_R \diff g \, f \bigg| - \big| S_R(f,g;P) \big|,
\end{equation}
so we conclude from Eq.~\eqref{Eq: epsilon inequality} that
\begin{equation}\label{Eq: bound with epsilon}
\bigg| \int_R \diff g \, f \bigg| \leq \norm{f}_{\infty} \operatorname{Var}_{R}\,(g) + \epsilon.
\end{equation}
Since Eq.~\eqref{Eq: bound with epsilon} holds for all $\epsilon>0$, this proves the claim.
\end{proof}

\begin{lemma}\label{Lemma: bounded variation}
Let $g:\R{n}\to\C{}$, and let $R= [x_1,y_1] \times \hdots \times [x_n,y_n]$ be a rectangle in $\R{n}$. If $g$ is of class $\mathcal{C}^{n}(R)$, then $g$ is of bounded variation in $R$, and
\begin{equation}
\operatorname{Var}_R\,(g) = \int_R \diff \bm z \, \bigg|\frac{\partial^n g}{\partial z_1\hdots\partial z_n}\bigg|.
\end{equation}
\end{lemma}
\begin{proof}
Consider a rectangle $K= [v_1,w_1] \times \hdots \times [v_n,w_n] \subset R$. Singling out the first coordinate, we can write
\begin{equation}
\Delta^{(n)}_{K}(g) = \sum_{\bm z} (-1)^{\# v(\bm z)} \big[ g(w_1, \bm z) - g(v_1,\bm z) \big] = \sum_{\bm z}(-1)^{\# v(\bm z)} g(w_1, \bm z) - \sum_{\bm z}(-1)^{\# v(\bm z)} g(v_1,\bm z),
\end{equation}
where $\bm z = (z_2,\hdots,z_{n}) \in \R{n-1}$, and $z_i=v_i$ or $w_i$ for each $i=2,\hdots,n$, with $\#v(\bm z)$ being as before the number of components $z_i$ of $\bm z$ that are equal to $v_i$. If we define
\begin{equation}
g_1(w) \coloneqq \sum_{\bm z}(-1)^{\# v(\bm z)} g(w, \bm z),
\end{equation}
then $g_1$ is continuously differentiable, and by the mean value theorem there exists $c_1 \in (v_1, w_1)$ such that
\begin{equation}\label{Eq: MVT first step}
\Delta^{(n)}_{K}(g) = g'_1(c_1) (w_1-v_1) = (w_1-v_1) \sum_{\bm z}(-1)^{\# v(\bm z)} \frac{\partial g}{\partial z_1}(c_1,\bm z).
\end{equation}
If $n=1$, we stop here. Otherwise, we can define \mbox{$K_1\coloneqq [v_2,w_2] \times \hdots \times [v_n,w_n]$}, and realize from Eq.~\eqref{Eq: MVT first step} that
\begin{equation}
\Delta^{(n)}_K(g) = (w_1-v_1) \,\Delta^{(n-1)}_{K_1}\bigg( \frac{\partial g}{\partial z_1}(c_1,\cdot) \bigg).
\end{equation}
Since $g$ is of class $\mathcal{C}^{n}(R)$, the partial derivative is of class $\mathcal{C}^{n-1}(R)$, and in particular so is its restriction to $K_1$. We can then repeat the previous argument with the partial derivative instead of $g$, concluding after this second step that there exists $c_2 \in (v_2,w_2)$ such that
\begin{equation}
\Delta^{(n)}_{K}(g) = (w_1-v_1)(w_2-v_2)\sum_{\bm z}(-1)^{\# v(\bm z)} \frac{\partial^2 g}{\partial z_1\partial z_2}(c_1,c_2,\bm z),
\end{equation}
where now $\bm z \in \R{n-2}$ runs only over the last $n-2$ components (if any). In general, after $n$ iterations we find that there exists $\bm c$ in the interior of $K$ such that
\begin{equation}\label{Eq: Generalized MVT}
\Delta^{(n)}_K(g) = \bigg[ \prod_{i=1}^{n} (w_i-v_i) \bigg] \frac{\partial^n g}{\partial z_1\hdots\partial z_n}(\bm c) =\operatorname{Vol}(K) \frac{\partial^n g}{\partial z_1\hdots\partial z_n}(\bm c).  
\end{equation}
Now, the variation of $g$ in $R$ is given by
\begin{equation}
\operatorname{Var}_R(g) = \sup_P \sum_{\bm j} |\Delta^{(n)}_{R_{\bm j}}g|,
\end{equation}
where the supremum is taken over all possible rectangular partitions of $R$. By Eq.~\eqref{Eq: Generalized MVT}, given an arbitrary partition $P$ of $R$, 
\begin{equation}
\sum_{\bm j} |\Delta^{(n)}_{R_{\bm j}}g| = \sum_{\bm j} \bigg| \frac{\partial^n g}{\partial z_1\hdots\partial z_n}(\bm c_{\bm j}) \bigg| \operatorname{Vol}(R_{\bm j}),
\end{equation}
where each $c_{\bm j}$ belongs to the interior of $R_{\bm j}$. This corresponds to the Riemann sum of \mbox{$\partial^n g / \partial z_1\hdots\partial z_n$} associated with the partition $P$ with tags $c_{\bm j}$. Since \mbox{$\partial^n g / \partial z_1\hdots\partial z_n$} is continuous, in particular it is Riemann integrable, and therefore given any $\epsilon > 0$, there exists $\delta>0$ such that if $\norm{P} < \delta$, then 
\begin{equation}
\bigg| \int_{R} \diff\bm z\, \bigg| \frac{\partial^n g}{\partial z_1\hdots\partial z_n}(\bm z) \bigg| - \sum_{\bm j} \bigg| \frac{\partial^n g}{\partial z_1\hdots\partial z_n}(\bm c_{\bm j}) \bigg| \operatorname{Vol}(R_{\bm j}) \bigg| < \epsilon. 
\end{equation}
In particular,
\begin{equation}
\int_{R} \diff\bm z\, \bigg| \frac{\partial^n g}{\partial z_1\hdots\partial z_n}(\bm z) \bigg| - \epsilon < \sum_{\bm j} \bigg| \frac{\partial^n g}{\partial z_1\hdots\partial z_n}(\bm c_{\bm j}) \bigg| \operatorname{Vol}(R_{\bm j}) \leq \operatorname{Var}_R(g).
\end{equation}
This is true for all $\epsilon > 0$, and thus
\begin{equation}\label{Eq: inequality 1}
\int_{R} \diff\bm z\, \bigg| \frac{\partial^n g}{\partial z_1\hdots\partial z_n}(\bm z) \bigg| \leq \operatorname{Var}_R(g).
\end{equation}
Conversely, for any partition $P$,
\begin{equation}
\sum_{\bm j} |\Delta^{(n)}_{R_{\bm j}}g| = \sum_{\bm j} \bigg| \int_{R_{\bm j}}\diff\bm z\, \frac{\partial^n g}{\partial z_1\hdots\partial z_n}(\bm z)   \bigg| \leq  \sum_{\bm j} \int_{R_{\bm j}} \diff\bm z\,\bigg|\frac{\partial^n g}{\partial z_1\hdots\partial z_n}(\bm z)   \bigg| = \int_{R}\diff\bm z\, \bigg|\frac{\partial^n g}{\partial z_1\hdots\partial z_n}(\bm z)   \bigg|.
\end{equation}
The supremum taken over all possible partitions still verifies the inequality, i.e, 
\begin{equation}\label{Eq: inequality 2}
\operatorname{Var}_R\,(g) \leq \int_{R}\diff\bm z\, \bigg|\frac{\partial^n g}{\partial z_1\hdots\partial z_n}(\bm z)   \bigg|.
\end{equation}
Combining Eqs.~\eqref{Eq: inequality 1} and~\eqref{Eq: inequality 2} gives us the desired result.
\end{proof}

\twocolumngrid	
\bibliography{references}
	
\end{document}